\documentclass{ws-ijgmmp}
\usepackage{enumitem}
\usepackage{amsfonts}
\usepackage{amssymb}
\usepackage{graphicx}
\usepackage{array}
\usepackage{appendix}
\usepackage{mathrsfs}
\usepackage{hyperref}
\usepackage{amsmath}
\numberwithin{equation}{section}
\newtheorem{thm}{Theorem}[section]

\newtheorem{prop}{Proposition}[section]
\newtheorem{defn}{Definition}[section]
\newtheorem{oss}{Remark}[section]

\begin{document}
\markboth{F. Alessio, G. Esposito}
{The Bondi-Mentzner-Sachs Group} 

\title{On the structure and applications of the Bondi-Mentzner-Sachs Group}

\author{Francesco Alessio}
\address{Universit\`a degli Studi di Napoli ``Federico II'', Dipartimento di Fisica ``Ettore Pancini'', 
Complesso Universitario di Monte
S. Angelo, Via Cintia Edificio 6, 80126 Napoli, Italy\\
\email{f.alessio@studenti.unina.it}}
\author{Giampiero Esposito}
\address{Istituto Nazionale di Fisica Nucleare, Sezione di
Napoli, Complesso Universitario di Monte S. Angelo, 
Via Cintia Edificio 6, 80126 Napoli, Italy\\
\email{gesposit@na.infn.it}}
\maketitle
\begin{abstract}
This work is a pedagogical review dedicated to a modern description of the 
Bondi-Metzner-Sachs (BMS) group. Minkowski space-time 
has an interesting and useful group of isometries, but, for a generic space-time, the isometry group 
is simply the identity and hence provides no significant informations. Yet symmetry groups have important 
role to play in physics; in particular, the Poincar\'e group, describing the isometries of Minkowski 
space-time plays a role in the standard definitions of energy-momentum and angular-momentum. 
For this reason alone it would seem to be important to look for a generalization of the concept of 
isometry group that can apply in a useful way to suitable curved space-times. The curved space-times 
that will be taken into account are the ones that suitably approach, at infinity, Minkowski space-time. 
In particular we will focus on asymptotically flat space-times. In this work the concept of 
\textit{asymptotic symmetry group} of those space-times will be studied. In the first two sections we 
derive the asymptotic group following the classical approach which was basically developed
by Bondi, van den Burg, Metzner and Sachs.
This is essentially the group of transformations between coordinate systems 
of a certain type in asymptotically flat space-times. In the third section the conformal method and the 
notion of \lq asymptotic simplicity\rq\hspace{0.1mm} are introduced, following mainly the works of 
Penrose. This section prepares us for another derivation of the BMS 
group which will involve the conformal structure, and is thus more geometrical and fundamental. 
In the subsequent sections we discuss the properties of the BMS group, e.g. 
its algebra and the possibility to obtain as its subgroup the Poincar\'e group, as we may expect. 
The paper ends with a review of the BMS invariance properties of classical gravitational 
scattering discovered by Strominger, that are finding application to black hole physics and quantum 
gravity in the literature.
\end{abstract}
\section{Introduction}
\label{sect:1}
Ever since Einstein developed his theory of general relativity, group-theoretical methods have played
an important role in deriving new solutions of the Einstein equations and understanding their
properties \cite{CUP2009,SA1,SA2,Carmeli}, and also in investigating the asymptotic structure of 
space-time \cite{Pen62,Pen63,Pen64,Pen65,Pen67}. In particular, shortly after that 
Bondi, Metzner and Sachs laid the foundations of the asymptotic
symmetry group (hereafter referred to as BMS) of asymptotically flat spacetime 
\cite{Bondi62,Sachs62,Sachs1}, McCarthy elucidated several
features of this group and its representations \cite{MC1,MC2,MC3,MC4,MC5}. For example, unlike the infinite-dimensional
representations of the Lorentz group, that allow for particles of arbitrary spin, a result first
obtained by Majorana \cite{Majo32}, McCarthy proved that the infinite-dimensional representations of the
BMS group only allow for discrete values of the spin of elementary particles \cite{MC1}.
Furthermore, over the last few years, the BMS group has been found to lead to new perspectives on
classical gravitational scattering \cite{ST2014,HPS1,HPS2} and on the problem of black-hole evaporation in
quantum gravity \cite{HPS1}. From a more mathematical point of view, all of this adds evidence in
favour of the pseudo-group structure of the functional equations of classical and quantum
physics being able to improve our understanding of the fundamental laws of nature. 

Within this conceptual framework, our review aims at introducing the general reader, who is not
necessarily a general relativist, to the modern way of understanding the BMS group and its applications.
For this purpose, we begin by recalling that the
importance of the concept of energy within a physical theory, if introduced correctly, arises from 
the fact that it is a conserved quantity in time and hence a very useful tool. 
Thus, in general relativity one of the most interesting questions is related to the meaning of gravitational energy.\\
Starting from any vector $J^a$ that satisfies a local conservation equation, that can be put in the form
\begin{equation}
\label{eqn:0}
\nabla_aJ^a=0,
\end{equation}
one can deduce an integral conservation law which states that the integral over the boundary 
$\partial\mathscr{D}$ of some compact region $\mathscr{D}$ of the flux of the vector $J^a$ across this 
boundary necessarily vanishes. In fact, using Gauss' theorem we have
\begin{equation}
\label{eqn:0.1}
\int_{\partial\mathscr{D}} J^ad\sigma_a=\int_{\mathscr{D}}\nabla_aJ^adv=0.
\end{equation}
Now we know that in General Relativity the energy-momentum tensor $T_{ab}$  satisfies the local conservation law
\begin{equation}
\label{eqn:1}
\nabla_aT^{ab}=0,
\end{equation}
which follows directly from the Einstein field equations. However from \eqref{eqn:1} we cannot deduce any integral conservation law. This is because in this case the geometric object to integrate over a 4-volume (as on the 
right-hand side of \eqref{eqn:0.1}) would be a vector and we can not take the sum of two vectors at different 
points of a manifold. This picture is ameliorated if space-time possesses symmetries, 
i.e. Killing vectors. If $K^a$ is a Killing vector,
\begin{equation*}
\nabla_{(a}K_{b)}=0,
\end{equation*}
we may build the vector 
\begin{equation*}
P^a=T^{ab}K_b,
\end{equation*}
that satisfies \eqref{eqn:0}, since
\begin{equation*}
\nabla_a P^a=\left(\nabla_aT^{ab}\right)K_b+T^{ab}\nabla_aK_b=0.
\end{equation*}
The second term vanishes because $T^{ab}$ is symmetric and so $T^{ab}\nabla_aK_b=T^{ab}\nabla_{(a}K_{b)}=0$.
Therefore the presence of Killing vectors for the metric leads to an integral conservation law. In  
flat Minkowski space-time we know that there are 10 Killing vectors:
\begin{equation*}
\textbf{L}_{\alpha}=\frac{\partial}{\partial x^{\alpha}},\hspace{2.2cm}(\alpha=0,1,2,3)
\end{equation*}
\begin{equation*}
\textbf{M}_{\alpha\beta}=e_{\alpha}x^{\alpha}\frac{\partial}{\partial x^{\beta}}
-e_{\beta}x^{\beta}\frac{\partial}{\partial x^{\alpha}},\hspace{2cm}
(\mathrm{no}\hspace{1.3mm}\mathrm{summation};\alpha,\beta =0,1,2,3)
\end{equation*}
where $e_{\alpha}$ is +1 if $\alpha=0$ and -1 if $\alpha=1,2,3$. The first four generate space-time translations 
and the second six \lq rotations\rq\hspace{0.1mm} in space-time (these are just the usual ten generators of 
the inhomogeneous Lorentz group). One may use them to define ten vectors $P^a_{\alpha}$ and 
$P^a_{\alpha\beta}$ which will obey \eqref{eqn:0}. We can think of $\textbf{P}_0$ as representing the flow of 
energy, and $\textbf{P}_1$, $\textbf{P}_2$ and $\textbf{P}_3$ as the flow of the three components of linear 
momentum. The $\textbf{P}_{\alpha\beta}$ can be interpreted as the flow of angular momentum. If the metric is 
not flat there will not, in general, be any Killing vectors. It is worth noting that the diffeomorphism group 
has, for historical reasons, frequently been invoked as a possible substitute for the Poincar\'e group for a 
generic space-time. However, it is not really useful in this context, being much too large and preserving only 
the differentiable structure of the space-time manifold rather than any of its physical more important properties.
\\ However, one could introduce in a suitable neighbourhood of a point $q$ normal coordinates $\{x^a\}$ so that 
the components $g_{ab}$ of the metric are $e_a\delta_{ab}$ (no summation) and that the components of 
$\Gamma^a{}_{bc}$ are zero at $q$. One may take a neighbourhood $\mathscr{D}$ of $q$ in which $g_{ab}$ and 
$\Gamma^a{}_{bc}$ differ from their values at $q$ by an arbitrary small amount. Then 
$\nabla_{(a}L_{\alpha\hspace{1mm}b)}$ and $\nabla_{(a}M_{\alpha\beta\hspace{1mm}b)}$ will not exactly vanish in 
$\mathscr{D}$, but will in this neighbourhood differ from zero by an arbitrary small amount. Thus\begin{equation*}
\int_{\partial\mathscr{D}}P^b_{\alpha}d\sigma_b\hspace{1cm}\mathrm{and}\hspace{1cm}
\int_{\partial\mathscr{D}}P^b_{\alpha\beta}d\sigma_b
\end{equation*}
will still be zero in the first approximation.
Hence the best we can get from \eqref{eqn:1} is an approximate integral conservation law, if we integrate 
over a region whose typical dimensions are very small compared with the radii of curvature involved in 
$R_{abcd}$. We can interpret this by regarding the space-time curvature as giving a non-local contribution 
to the energy-momentum, that has to be considered in order to obtain a correct integral conservation law.\\
From the above discussion we deduce that no exact symmetries can be found for a generic space-time. 
However, if we turn to the concept of asymptotic symmetries and we apply it to asymptotically flat space-times, 
we will see that the picture is not so bad and that we can still talk about the Poincar\'e group. The basic 
idea, developed in the remainder of the paper, is that, since we are taking into account asymptotically 
flat space-times, we may expect that by going to \lq infinity\rq\hspace{0.1mm} one might acquire the Killing 
vectors necessarily for stating integral conservation laws. 
\section{Bondi-Sachs coordinates and Boundary Conditions} 
\label{sect:2}
Consider the Minkowski metric
\begin{equation*}
g=\eta_{ab}dx^a\otimes dx^b=dt\otimes dt-dx\otimes dx-dy\otimes dy-dz\otimes dz.
\end{equation*}
We introduce new coordinates 
\begin{equation}
\label{eqn:135}
u=t-r,\hspace{1cm}r\cos\theta=z,\hspace{1cm}r\sin\theta e^{i\phi}=x+iy,
\end{equation}
in terms of which the Minkowski metric takes the form
\begin{equation}
\label{eqn:137}
g=du\otimes du+du\otimes dr+dr\otimes du-r^2(d\theta\otimes d\theta+\sin^2\theta d\phi\otimes d\phi),
\end{equation}
which can also be written as
\begin{equation}
\label{eqn:139}
ds^2=du\otimes du+du\otimes dr+dr\otimes du-r^2q_{AB}dx^A\otimes dx^B, 
\end{equation}
where 
\begin{equation*}
q_{AB}=\left(\begin{matrix} 1& 0\\ 0& \sin^2\theta\end{matrix}\right),\hspace{0.5cm}A,B,...=2,3.
\end{equation*}
Note that $q_{AB}$ represents the metric on the unit $2$-sphere. 
The coordinate $u$ is called \textit{retarded time}.\\
We proceed to the interpretation of the coordinates \eqref{eqn:135}. The hypersurfaces given by the equation 
$u=\mathrm{const}$ are null hypersurfaces, since their normal co-vector $k_a=\nabla_a u$ is null. They are 
everywhere tangent to the light-cone. Note that it is a peculiar property of null hypersurfaces that their 
normal direction is also tangent to the hypersurface. The coordinate $r$ is such that the area of the surface 
element $u=\mathrm{const}$, $r=\mathrm{const}$ is $r^2\sin\theta d\theta d\phi$. Define a ray as the line with 
tangent $k^a=g^{ab}\nabla_bu$. Then the scalars $\theta$ and $\phi$ are constant along each ray.\\
Now we would like to introduce for a generic metric tensor a set of coordinates $(u,r,x^A)$ which has the same 
properties as the ones of \eqref{eqn:135}. These coordinates are known as \textit{Bondi-Sachs coordinates} 
\cite{Bondi62,Sachs62,Sachs1}. The hypersurfaces $u=\mathrm{const}$ are null, i.e. the normal co-vector 
$k_a=\nabla_au$ satisfies $g^{ab}(\nabla_au)(\nabla_bu)=0$, so that $g^{uu}=0$, and the corresponding 
future-pointing vector $k^a=g^{ab}\nabla_bu$ is tangent to the null rays. Two angular coordinates $x^A$, 
with $A,B,...=2,3$, are constant along the null rays, i.e. $k^a\nabla_ax^A=g^{ab}(\nabla_au)\nabla_bx^A=0$, 
so that $g^{uA}=0$. The coordinate $r$, which varies along the null rays, is chosen to be an areal coordinate 
such that $\mathrm{det}[g_{AB}]=r^4\mathrm{det}[q_{AB}]$, where $q_{AB}$ is the unit sphere metric associated 
with the angular coordinates $x^A$, e.g. $q_{AB}=\mathrm{diag}(1,\sin^2\theta)$ for standard spherical coordinates 
$x^A=(\theta,\phi)$. The contravariant components $g^{ab}$ and covariant components $g_{ab}$ are related by 
$g^{ac}g_{cb}=\delta^a_b$, which in particular implies $g_{rr}=0$ (from $\delta_{ur}=0$) and $g_{rA}=0$ 
(from $\delta_{uA}=0$). See Figure \ref{fig:3.1}.
\begin{figure}[h]
\begin{center}
\includegraphics[scale=0.52]{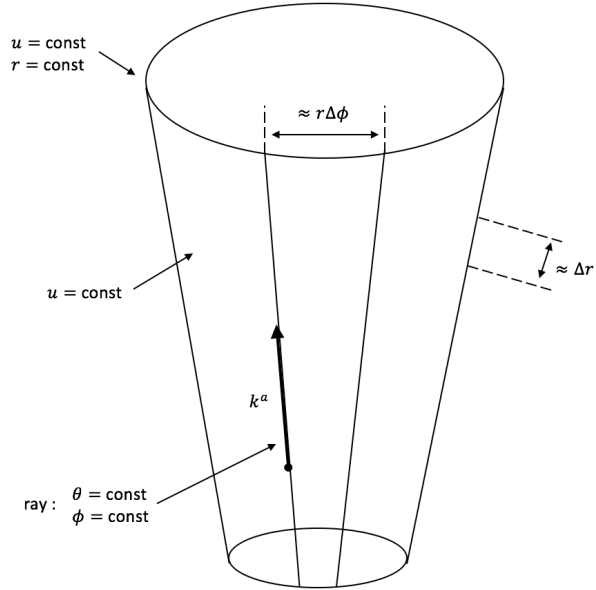}
\caption{The Bondi-Sachs coordinate system. The coordinates $u$, $r$, and $\phi$ and the vector $k^a$ are shown 
in the hypersurface $\theta=\mathrm{const}$.}
\label{fig:3.1}
\end{center}
\end{figure}
\\It can be shown \cite{Bondi62} that the metric tensor takes the form 
\begin{eqnarray}
\label{eqn:136}
g&=& g_{ab}dx^a\otimes dx^b=e^{2\beta}\frac{V}{r}du\otimes du+e^{2\beta}(du\otimes dr+dr\otimes du)
\nonumber \\
&+& g_{AB}(dx^A-U^Adu)\otimes(dx^B-U^Bdu),
\end{eqnarray}
where
\begin{equation}
\label{eqn:168}
g_{AB}=r^2h_{AB},\hspace{1cm}\mathrm{det}[h_{AB}]=h(x^A).
\end{equation}
Using Jacobi's formula for the derivative of a determinant for a generic matrix $g_{\mu\nu}$,
\begin{equation*}
 \partial_{\rho}\mathrm{det}[g_{\mu\nu}]=\partial _{\rho}g=gg^{\mu\nu}\partial_{\rho}g_{\mu\nu},
\end{equation*}
we have from the second of \eqref{eqn:168}
\begin{equation}
\label{eqn:169}
\partial_u h=0\Rightarrow h^{AB}\partial_uh_{AB}=0,\hspace{1cm}\partial_r h=0\Rightarrow h^{AB}\partial_rh_{AB}=0.
\end{equation}
We also have
\begin{equation*}
g^{ur}=e^{-2\beta},\hspace{0.5cm}g^{rr}=-\frac{V}{r}e^{-2\beta},\hspace{0.5cm}g^{rA}
=U^Ae^{-2\beta},\hspace{0.5cm}g^{AB}=-\frac{1}{r^2}h^{AB}.
\end{equation*}
A suitable representation for $h_{AB}$ is the following:\\
\begin{equation}
\label{eqn:141}
h_{AB}=\left(\begin{matrix}\cosh2\delta e^{2\gamma} & \sin\theta\sinh2\delta \\ \\ \\\sin\theta\sinh2\delta 
& \sin^2\theta\cosh2\delta e^{-2\gamma}\end{matrix}\right)\Rightarrow \mathrm{det}[h_{AB}]=\sin^2\theta.
\end{equation}
Here $V$, $\beta$, $U^A$, $\gamma$ and $\delta$ are any six functions of the coordinates. The form 
\eqref{eqn:136} holds if and only if $(u,r,\theta,\phi)$ have the properties stated above. Note that this 
form differs from the original form of Sachs \cite{Sachs62} by the transformation $\gamma\rightarrow(\gamma+\delta)/2$ 
and $\delta\rightarrow(\gamma-\delta)/2$. The original axisymmetric Bondi metric \cite{Bondi62} with rotational 
symmetry in the $\phi$-direction was characterized by $\delta=U^{\phi}=0$ and $\gamma=\gamma(u,r,\theta)$, 
resulting in a metric with reflection symmetry $\phi\rightarrow-\phi$ so that it is not suitable for describing 
an axisymmetric rotating star. \\
The next step is to write down the Einstein vacuum field equations in the above coordinate system in order to find 
the equations that rule the evolution of the six arbitrary functions on which the metric depends. 
As shown in \cite{Sachs62} or \cite{BSF} the Einstein vacuum field equations
\begin{equation*}
G_{ab}=R_{ab}-\frac{1}{2}Rg_{ab}=0,
\end{equation*}
separate into the \textit{Hypersurface equations},
\begin{equation*}
G^u_a=0,
\end{equation*}
and the \textit{Evolution equations},
\begin{equation*}
G_{AB}-\frac{1}{2}g_{AB}g^{CD}G_{CD}=0.
\end{equation*}
The former determines $\beta$ along the null rays ($G^u_r=0$), $U^A$ ($G^u_A=0$) and $V$ ($G^u_u=0$), 
while the latter gives informations about the retarded time derivatives of the two degrees of freedom 
contained in $h_{AB}$. Usually one requires the following conditions:
\begin{enumerate}
\item For any choice of $u$ one can take the limit $r\rightarrow\infty$ along each ray;
\item For some choice of $\theta$ and $\phi$ and the above choice of $u$ the metric \eqref{eqn:136} 
should approach the Minkowski metric \eqref{eqn:137}, i.e. 
\begin{equation}
\label{eqn:138}
\lim_{r\to\infty}\beta=\lim_{r\to\infty}U^A=0,\hspace{0.5cm}\lim_{r\to\infty}\frac{V}{r}
=1,\hspace{0.5cm}\lim_{r\to\infty}h_{AB}=q_{AB}.
\end{equation}
Note that these conditions, as pointed out in \cite{Sachs62}, are rather unsatisfactory from a geometrical 
point of view. They will be completely justified later, using the method of the conformal structure, introduced by Penrose;
\item Over the coordinate ranges $u_0\leq u\leq u_1$, $r_0\leq r\leq\infty$, $0\leq\theta\leq\pi$ 
and $0\leq\phi\leq 2\pi$ all metric functions can be expanded in series of $r^{-1}$. 
\end{enumerate}
Using the Einstein equations with these assumptions it can be shown \cite{Sachs62,BSF} that the 
following asymptotic behaviours hold:
\begin{subequations}
\label{eqn:145}
\begin{align}
&V=r-2M+O(r^{-1}),\\
&h_{AB}=q_{AB}+\frac{c_{AB}}{r}+O(r^{-2}),\\
&\beta=-\frac{c^{AB}c_{AB}}{32r^2}+O(r^{-3}),\\
&U^A=-\frac{D_{B}c^{AB}}{2r^2}+O(r^{-3}),
\end{align}
\end{subequations}
i.e. the metric \eqref{eqn:136} admits the asymptotic expansion
\begin{eqnarray}
\label{eqn:140}
g&=& du\otimes du+du\otimes dr+dr\otimes du-r^2q_{AB}dx^A\otimes dx^B
\nonumber \\
&-& \frac{2m_{_{B}}}{r}du\otimes du-\frac{c^{AB}c_{AB}}{4r^2}(du\otimes dr+dr\otimes du)
\nonumber \\
&-& rc_{AB}dx^A\otimes dx^B-\frac{D_Fc^{F}_{A}}{2}(du\otimes dx^A+dx^A\otimes du)+...
\end{eqnarray}
Here the function $m_{_{B}}=m_{_{B}}(u,\theta,\phi)$ is called the \textit{Bondi mass aspect}, $c_{AB}=c_{AB}(u,\theta,\phi)$ 
represents the $O(r^{-1})$ correction to $h_{AB}$ and $D_A$ is the covariant derivative with respect to the metric 
on the unit $2$-sphere, $q_{AB}$ \cite{Haw17}. Capital letters A, B,... can be raised and lowered with respect 
to $q_{AB}$. In carrying out the $1/r$ expansion of the field equations the covariant derivative ${\cal D}_A$ corresponding 
to the metric $h_{AB}$ is related to the covariant derivative $D_A$ 
corresponding to the unit sphere metric $q_{AB}$ by 
\begin{subequations}
\label{eqn:160}
\begin{equation}
{\cal D}_AV^B=D_AV^B+C^B{}_{AE}V^E,
\end{equation}
where
\begin{equation}
C^B{}_{AE}=\frac{1}{2r}q^{BF}\Bigr(D_Ac_{FE}+D_Ec_{FA}-Dc_{AE}\Bigr)
+O(r^{-2}).
\end{equation}
\end{subequations}
This property will be useful later.
\begin{defn}$\\ $
\label{defn:AF1}
A space-time $(\mathscr{M},g)$ is \textit{asymptotically flat} if the metric tensor $g$ and its components 
satisfy the conditions \eqref{eqn:145} and \eqref{eqn:140}. These conditions are often 
referred to as \textit{boundary conditions}.
\end{defn} 
\section{Bondi-Metzner-Sachs group}
\label{sect:3}
In this section our purpose is to find the coordinate transformations which preserve the asymptotic 
flatness condition. In other words we want to find the asymptotic isometry group of the metric \eqref{eqn:136} and we must demand some conditions to hold in order for the coordinate conventions and boundary conditions 
to remain invariant. It is clear that, from \eqref{eqn:140}, the corresponding changes suffered from the metric 
must therefore obey certain fall-off conditions, i.e.
\begin{equation}
\label{eqn:146}
\delta g_{rr}=0,\hspace{1cm}\delta g_{rA}=0,\hspace{1cm} g^{AB}\delta g_{AB}=0.
\end{equation} 
and 
\begin{subequations}
\label{eqn:181}
\begin{equation}
\label{eqn:147}
\delta g_{uu}=O(r^{-1}),\hspace{1cm}\delta g_{uA}=O(1),
\end{equation} 
\begin{equation}
\label{eqn:180}
\delta g_{ur}=O(r^{-2}),\hspace{1cm}\delta g_{AB}=O(r).
\end{equation} 
\end{subequations}
The third of \eqref{eqn:146} expresses the fact that we \textit{don't} want the angular metric $g_{AB}$ to undergo 
any conformal rescaling under the transformation. However a generalization which includes conformal rescalings 
of $g_{AB}$ can be found in \cite{Barn2010}.\\
We know that the infinitesimal change $\delta g_{ab}$ in the metric tensor is given by the Lie derivative of 
the metric along the $\xi^a$ direction, $\xi^a$ being the generator of the transformation of coordinates:
\begin{equation}
\label{eqn:148}
\delta g_{ab}=-\nabla_a\xi_b-\nabla_b\xi_a.
\end{equation}
Clearly the vector $\xi^a$ obeys \textit{Killing's equation},
\begin{equation*}
\nabla_a\xi_b+\nabla_b\xi_a=0,
\end{equation*}
if and only if the corresponding transformations are isometries. What we want to solve now is an 
\textit{asymptotic Killing's equation}, obtained putting together \eqref{eqn:146} and \eqref{eqn:181} 
with \eqref{eqn:148}. We get from the first of \eqref{eqn:146} 
\begin{equation*}
\nabla_r\xi_r=\partial_r\xi_r-\Gamma^{u}{}_{rr}\xi_u-\Gamma^{r}{}_{rr}\xi_r-\Gamma^{A}{}_{rr}\xi_A=0,
\end{equation*}
and using the Christoffel symbols given in \ref{C}, we get 
\begin{equation*}
\partial_r\xi_r=2\partial_r\beta,
\end{equation*}
and hence
\begin{equation}
\label{eqn:143}
\xi_r=f(u,x^A)e^{2\beta},
\end{equation}
where $f$ is a suitably differentiable function of its arguments.\\
From the second of \eqref{eqn:146} we obtain
\begin{equation*}
\nabla_r\xi_A+\nabla_A\xi_r=\partial_r\xi_A+\partial_A\xi_r-2\Gamma^{u}{}_{rA}
\xi_u-2\Gamma^{r}{}_{rA}\xi_r-2\Gamma^{B}{}_{rA}\xi_B=0,
\end{equation*}
and thus, using \eqref{eqn:143} we get 
\begin{equation*}
\partial_r\xi_A-r^2h_{AB}f\left(\partial_rU^B\right)-\frac{2\xi_A}{r}
-\left(\partial_rh_{AC}\right)h^{BC}\xi_B=-\left(\partial_Af\right)e^{2\beta},
\end{equation*}
and after some manipulation
\begin{equation*}
\partial_r\left(\xi_Bg^{BD}+fU^D\right)=-e^{2\beta}g^{AD}\left(\partial_Af\right),
\end{equation*}
which leads to 
\begin{equation*}
\xi_A=-h_{DA}f^Dr^2+fU^Dh_{DA}r^2+r^2h_{DA}\left(\partial_Bf\right)\int_r^{\infty}\frac{e^{2\beta}h^{BD}}{r'^2}dr'
\end{equation*}
\begin{equation}
\label{eqn:144}
=-f_Ar^2+fU_Ar^2+I_Ar^2+O(r),
\end{equation}
where 
\begin{equation*}
I^D(u,r,x^A)=(\partial_Bf)\int_r^{\infty}\frac{e^{2\beta}h^{BD}}{r'^2}dr'=\frac{\partial^Df}{r}+O(r^{-2}),
\end{equation*}
where $f^D$ are suitably differentiable functions of their arguments and the indices A, B etc. 
are raised and lowered with respect to the metric $q_{AB}$.\\ We can solve algebraically 
the third equation in \eqref{eqn:146} to obtain $\xi_u$:
\begin{equation*}
\xi_u=-\frac{e^{2\beta}}{2r}\left(-\partial_a\xi_B+\Gamma^{r}{}_{AB}\xi_r+\Gamma^{C}{}_{AB}\xi_C\right)h^{AB}.
\end{equation*}
Working with Christoffel symbols we get the following expression for $\xi_u$:
\begin{eqnarray}
\label{eqn:149}
\xi_u &=& -\frac{e^{2\beta}r}{4}\partial_D\left(h_{AB}f^D\right)h^{AB}+\frac{e^{2\beta}r}{2}
\left(\partial_Af\right)U^{A}-\frac{e^{2\beta}r}{4}\partial_D\left(h_{AB}I^D\right)h^{AB}
\nonumber \\
&+& e^{2\beta}\frac{V}{r}+r^2h_{AB}(U^Af^B-r^2U^AU^Bf-r^2U^AI^B).
\end{eqnarray}
Now equations \eqref{eqn:181} can be used to give constraints on the arbitrary functions 
$f$ and $f^A$. From the second of \eqref{eqn:180} we get
\begin{equation*}
\nabla_A\xi_B+\nabla_B\xi_A=\partial_A\xi_B+\partial_B\xi_A-2\Gamma^{u}{}_{AB}
\xi_u-2\Gamma^{r}{}_{AB}\xi_r-2\Gamma^C{}_{AB}\xi_C=O(r).
\end{equation*}
Using asymptotic expansions \eqref{eqn:145}, taking the order $r^2$ of the previous equation and putting it equal to zero we get
\begin{equation*}
-\partial_Af_B-\partial_Bf_A+\frac{1}{2}q_{AB}\partial_D\left(q_{CE}f^D\right)q^{CE}+q^{CD}
\left(\partial_Aq_{DB}+\partial_Bq_{DA}-\partial_Dq_{AB}\right)f_C=0,
\end{equation*}
thus
\begin{equation*}
-\partial_Af_B+\gamma^{C}{}_{AB}f_C-\partial_Af_B+\gamma^{C}_{AB}f_C=-\frac{1}{2}q_{AB}\partial_D\left(q_{CE}\right)q^{CE},
\end{equation*}
where $\gamma^A{}_{BC}$ are the Christoffel symbols with respect to the metric on the unit sphere $q_{AB}$. We eventually get
\begin{equation}
\label{eqn:150}
D_Af_B+D_Bf_A=\frac{1}{2}q_{AB}\partial_D\left(q_{CE}f^D\right)q^{CE}.
\end{equation}
and hence
\begin{equation*}
D_Af_B+D_Bf_A=f^D\frac{1}{2}q_{AB}\left(\partial_Dq_{CE}\right)q^{CE}+\left(\partial_Df^D\right)q_{AB}=q_{AB}D_Cf^C.
\end{equation*}
Thus $f^B$ are the conformal Killing vectors of the unit 2-sphere metric $q_{AB}$.\\
From the second of \eqref{eqn:147} we get 
\begin{equation*}
\nabla_u\xi_A+\nabla_A\xi_u=\partial_u\xi_A+\partial_A\xi_u-2\Gamma^{u}{}_{Au}
\xi_u-2\Gamma^{r}{}_{Au}\xi_r-2\Gamma^{B}{}_{Au}\xi_B=O(1).
\end{equation*}
Putting the order $r^2$ of this equation equal to zero we obtain
\begin{equation}
\label{eqn:151}
\partial_uf_A=0. 
\end{equation}
From the first of \eqref{eqn:180} we get
\begin{equation*}
\nabla_u\xi_r+\nabla_r\xi_u=\partial_u\xi_r+\partial_r\xi_u-2\Gamma^{u}_{ur}
\xi_u-2\Gamma^{r}_{ur}\xi_r-2\Gamma^{A}_{ur}	\xi_A=O(r^{-2}).
\end{equation*}
Putting the term of order $r^0$ of the previous equation equal to zero we get 
\begin{equation}
\label{eqn:152}
\partial_uf=\frac{1}{4}\partial_D\left(q_{AB}f^D\right)q^{AB}.
\end{equation}
Putting all the results together we have
\begin{subequations}
\label{eqn:153}
\begin{equation}
\partial_uf_A=0\Rightarrow f_A=f_A(x^B),\end{equation}
\begin{align}
D_Af_B+D_Bf_A=2q_{AB}\partial_uf\Rightarrow \left\{\begin{matrix} \partial^2_uf=0,&\\
\partial_uf=\frac{1}{2}D_Af^A.&\end{matrix}\right.
\end{align}
\end{subequations}
We get for $f$ the following expansion
\begin{equation}
f=\alpha+\frac{u}{2}D_Af^A,
\end{equation}
where $\alpha$ is a suitably differentiable function of $x^A$.\\
Consider now
\begin{equation*}
\xi^a=g^{ab}\xi_b,
\end{equation*}
from which we get 
\begin{equation}
\label{eqn:154}
\xi^u=f=\alpha+\frac{u}{2}D_Af^A,
\end{equation}
\begin{equation}
\label{eqn:155}
\xi^A=f^A-I^A=f^A-\frac{D^A\alpha}{r}-u\frac{D^AD_Cf^C}{2r}+O(r^{-2}),
\end{equation}
\begin{eqnarray}
\label{eqn:161}
\xi^r &=& -\frac{r}{2}\left[D_A\xi^A-U^A\partial_Af\right]=-\frac{r}{2}D_{C}\xi^C+O(r^{-1})
\nonumber \\
&=& -\frac{r}{2}D_Cf^C+\frac{D_CD^C\alpha}{2}+u\frac{D_CD^CD_Af^A}{4}+O(r^{-1}).
\end{eqnarray}
The second equality in \eqref{eqn:161} follows from \eqref{eqn:160} and from 
\begin{equation*}
q^{AB}c_{AB}=0,
\end{equation*}
which follows from satisfying at order $r^{-2}$ the second of \eqref{eqn:169} in the form 
\begin{equation*}
0=h^{AB}\partial_rh_{AB}=[q^{AB}-\frac{c^{AB}}{r^2}+O(r^{-3})][-\frac{c_{AB}}{r^2}+O(r^{-3})].
\end{equation*}
As $r\rightarrow\infty$  \eqref{eqn:154} and \eqref{eqn:155} become, respectively
\begin{align}
\label{eqn:156}
&\xi^u=\alpha+\frac{u}{2}D_Af^A,\\
\label{eqn:157}
&\xi^A=f^A.
\end{align}
Finally we can state that the asymptotic Killing vector is of the form \begin{equation}
\label{eqn:158}
\xi=\xi^a\partial_a=\left[\alpha(x^C)+\frac{u}{2}D_Af^A(x^C)\right]\partial_u+f^A(x^C)\partial_A,
\end{equation}
where $\alpha$ is arbitrary and $f^A$ are the conformal Killing vectors of the metric of the unit sphere. In order to fix ideas, set $x^A=(\theta,\phi)$.
It is clear then that $\theta$ and $\phi$ undergo a finite conformal transformation, i.e. 
\begin{subequations}
\label{eqn:159}
\begin{equation}
\label{eqn:182}
\theta\rightarrow\theta'=F(\theta,\phi),
\end{equation}
\begin{equation}
\label{eqn:183}
\phi\rightarrow\phi'=G(\theta,\phi),
\end{equation}
for which
\begin{equation*}
d\theta'^2+\sin^2\theta' d\phi'^2=K^2(\theta,\phi)(d\theta^2+\sin^2\theta d\phi^2),
\end{equation*}
and hence
\begin{equation}
\label{eqn:171}
K^4=J^{2}(\theta,\phi;\theta',\phi')\sin^2\theta\left(\sin\theta'\right)^{-2},\hspace{0.5cm}J
=\mathrm{det}\left(\begin{matrix}\frac{\partial F}{\partial\theta} & \frac{\partial F}
{\partial\phi}\\\frac{\partial G}{\partial\theta} & \frac{\partial G}{\partial\phi}
\end{matrix}\right).
\end{equation}
By definition of conformal Killing vector we also have
\begin{equation}
\label{eqn:214}
K^2=e^{D_Af^A}.
\end{equation}
The finite form of the transformation of the coordinate $u$ is given, as can be easily checked, by 
\begin{equation}
\label{eqn:184}
u\rightarrow u'=K[u+\alpha(\theta,\phi)].
\end{equation}
\end{subequations}
\begin{defn}$\\ $
\label{defn:BMS}
The transformations \eqref{eqn:159} are called \textit{BMS (Bondi-Metzner-Sachs) transformations}, and 
are the set of diffeomorphisms which leave the asymptotic form of the metric of 
an asymptotically flat space-time unchanged.
\end{defn}
The BMS transformations form a group. In fact, as is known, the conformal transformations form a group, so that 
$F$, $G$, and $K$ have all the necessary properties. Thus, one must only check the fact that if one 
carries out two transformation \eqref{eqn:184} the corresponding $\alpha$ for the product is again 
a suitably differentiable function of $\theta$ and $\phi$. If
\begin{equation*}
u_1\rightarrow u_2=K_{12}[u_1+\alpha_{12}]
\end{equation*} 
and 
\begin{equation*}
u_2\rightarrow u_3=K_{23}[u_2+\alpha_{23}]
\end{equation*}
then we have 
\begin{equation*}
u_1\rightarrow u_3=K_{13}[u_1+\alpha_{13}],\hspace{0.5cm}K_{13}=K_{12}K_{23},\hspace{0.5cm}
\alpha_{13}=\alpha_{12}+\frac{\alpha_{23}}{K_{12}}.
\end{equation*}
Since $\alpha_{13}$ is a suitably differentiable function it follows that 
\begin{prop}$\\ $
The BMS transformations form a group, denoted with $\mathscr{B}$.
\end{prop}
\begin{defn}$\\ $
The BMS transformations for which the determinant $J$, defined in \eqref{eqn:171}, 
is positive form the \textit{proper} subgroup of the BMS group.
\end{defn}
In the remainder we will omit the word \lq proper\rq, even if all of our considerations will regard this component of $\mathscr{B}$.
\begin{oss}$\\ $
Note that the $r$ coordinate too may be involved in the BMS group of transformations, but such a 
transformation is somewhat arbitrary since it depends on the precise type of radial coordinate used and 
it is not relevant to the structure of the group. Clearly the BMS group is infinite-dimensional since 
the transformations depend upon a suitably differentiable function $\alpha(\theta,\phi)$.
\end{oss}
\section{Conformal Infinity}
\label{sect:4}
In this part of the work the notion of conformal infinity, originally introduced by Penrose, is developed. 
The idea is that if the space-time is considered from the point of view of its conformal structure only, 
\lq points at infinity\rq\hspace{0.1mm} can be treated on the same basis as finite points. This can be 
done completing the space-time manifold to a highly symmetrical conformal manifold by the addition of 
a null cone at infinity, called $\mathscr{I}$. We want to construct \cite{Pen62,Pen63,Pen64,Pen67}, 
starting from the \lq physical space-time\rq\hspace{0.1mm} $(\mathscr{\tilde M},\tilde{g})$, another 
\lq unphysical space-time\rq\hspace{0.1mm} $(\mathscr{M},g)$ with boundary $\mathscr{I}
=\mathscr{\partial M}$ \cite{Lang}, such that $\mathscr{\tilde M}$ is conformally equivalent to the 
interior of $\mathscr{M}$ with $g_{ab}=\Omega^2\tilde g_{ab}$, given an appropriate function $\Omega$. 
The two metrics $\tilde g_{ab}$ and $g_{ab}$ define on $\mathscr{\tilde{M}}$ the same null-cone structure. 
The function $\Omega$ has to vanish on $\mathscr{I}$, so that the physical metric would have to be infinite 
on it and cannot be extended. The boundary $\mathscr{I}$ can be thought as being at infinity, in the sense 
that any affine parameter in the metric $\tilde{g}$ on a null geodesic in $\mathscr{M}$ attains unboundedly 
large values near $\mathscr{I}$. This is because if we consider an affinely parametrized null geodesic 
$\gamma$ in the unphysical space-time $(\mathscr{M},g)$ with affine parameter $\lambda$, whose equation is\begin{equation*}
\frac{d^2x^a}{d\lambda^2}+\Gamma^a{}_{bc}\frac{dx^b}{d\lambda}\frac{dx^c}{d\lambda}=0,
\end{equation*}
it is easy to see that the corresponding geodesic $\tilde{\gamma}$ in the physical space-time 
$(\tilde{\mathscr{M}},\tilde{g})$ with affine parameter $\tilde{\lambda}(\lambda)$ is solution of the equation 
\begin{equation*}
\frac{d^2x^a}{d\tilde{\lambda}^2}+\tilde{\Gamma}^a{}_{bc}\frac{dx^b}{d\tilde{\lambda}}
\frac{dx^c}{d\tilde{\lambda}}=-\frac{1}{\tilde{\lambda}'}\left(\frac{\tilde{\lambda}''}
{\tilde{\lambda}'}+2\frac{\Omega'}{\Omega}\right)\frac{dx^a}{d\tilde{\lambda}},
\end{equation*} 
where $'$ denotes a $\lambda$ derivative. If we want the parameter $\tilde{\lambda}$ to be affine the 
right-hand side of the above equation must vanish, and hence we must have
\begin{equation*}
\frac{d\tilde{\lambda}}{d\lambda}=\frac{c}{\Omega^2},
\end{equation*}
where $c$ is an arbitrary constant. Since $\Omega=0$ on $\mathscr{I}$, $\tilde{\lambda}$ diverges and hence 
$\tilde{\gamma}$ never reaches $\mathscr{I}$, which apparently really is at infinity. Thus, from the point 
of view of the physical metric, the new points (i.e. those on $\mathscr{I}$) are infinitely distant from 
their neighbours and hence, physically, they represent \lq points at infinity\rq .\\ The advantage in studying 
the space-time $(\mathscr{M},g)$ instead of $(\tilde{\mathscr{M}},\tilde{g})$ is that the infinity of the 
latter gets represented by a finite hypersurface $\mathscr{I}$ and the asymptotic properties of the fields 
defined on it can be investigated by studying $\mathscr{I}$ and the  behaviour of such fields on $\mathscr{I}$. \\
However, there is a large freedom for the choice of the function $\Omega$. Anyway, it turns out \cite{Penrin2} 
from general considerations that an appropriate behaviour for $\Omega$ is that it should approach zero 
(both in the past and in the future) like the reciprocal of an affine parameter $\lambda$ on a null geodesic 
of the space-time considered ($\lambda\Omega\rightarrow \mathrm{constant}$ as $\lambda\rightarrow\pm\infty$).
\\Consider physical Minkowski space-time in spherical polar coordinates
\begin{equation}
\tilde{g}=dt\otimes dt -dr\otimes dr-r^2\Sigma_2,
\end{equation}
where
\begin{equation}
\label{eqn:sfer}
\Sigma_2=d\theta\otimes d\theta+\sin^2\theta d\phi\otimes d\phi.
\end{equation}
Introduce now the standard retarded and advanced null coordinates $(t,r)\rightarrow(u,v)$ defined by
\begin{equation*}
u=t-r,\hspace{1cm}v=t+r,\hspace{1cm} v\geq u.
\end{equation*}
The coordinates $u$ and $v$ serve as affine parameters into the past and into the future of 
null geodesics of Minkowski space-time.\\
The metric tensor becomes
\begin{equation*}
\tilde{g}=\frac{1}{2}(du\otimes dv+dv\otimes du)-\frac{1}{4}(v-u)^2\Sigma_2.
\end{equation*}
Consider now the unphysical metric
\begin{equation*}
g=\Omega^2\tilde{g},
\end{equation*}
with the choice
\begin{equation*}
\Omega^2=\frac{4}{(1+u^2)(1+v^2)}.
\end{equation*}
Note that for $u,v\rightarrow\pm\infty$  we have $\Omega u$, $\Omega v\rightarrow \mathrm{constant}$, 
as pointed out before.\\ Now to interpret this metric it is convenient to introduce new coordinates
\begin{equation*}
u=\tan p,\hspace{1cm}v=\tan q,\hspace{1cm}-\frac{\pi}{2}<p\leq q<\frac{\pi}{2},
\end{equation*}
such that we have
\begin{equation}
\label{eqn:2}
g=2(dp\otimes dq+dq\otimes dp)-\sin^2(p-q)\Sigma_2.
\end{equation}
It is possible to bring the metric \eqref{eqn:2} in a more familiar form by setting
\begin{equation*}
t'=q+p,\hspace{0.8cm}r'=q-p,\hspace{0.8cm}-\pi<t'<\pi,\hspace{0.8cm}-\pi<t'-r'<\pi,\hspace{0.8cm}0<r'<\pi,
\end{equation*}
from which follows 
\begin{equation}
\label{eqn:3}
g=dt'\otimes dt'-dr'\otimes dr'-\sin^2 (r')\Sigma_2.
\end{equation}
It is worth noting that the metric \eqref{eqn:3} is that of \textit{Einstein's static universe}, $\mathscr{E}$, 
the cylinder obtained as product between the real line and the 3-sphere, $S^3\times\mathbb{R}$. However, the 
manifold $\mathscr{M}$ represents just a finite portion of such a cylinder.
\begin{figure}[h]
\begin{center}
\includegraphics[scale=0.29]{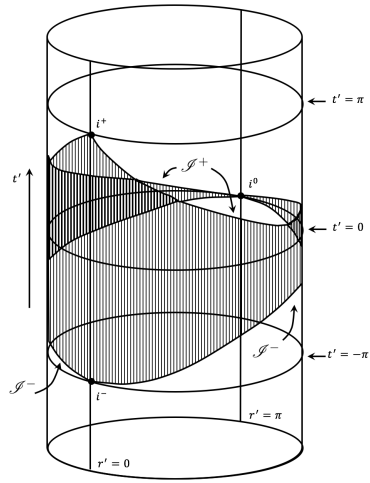}
\caption{The cylinder $\mathscr{E}=S^3\times\mathbb{R}$, of which $\mathscr{M}$ is just a finite portion, 
delimited by $\mathscr{I^+}$, $\mathscr{I^-}$, $i^+$, $i^-$ and $i^0$.
We note that the $(\theta,\phi)$ coordinates are suppressed, so that each point 
represents a 2-sphere of radius $\sin r'$.}
\label{fig:2.3}
\end{center}
\end{figure}
\\The metric \eqref{eqn:2} is defined at $q=\pi/2$ and $p=-\pi/2$: those values correspond to the infinity 
of $\mathscr{\tilde M}$ and therefore they represent the hypersurface $\mathscr{I}$. Hence we have defined 
a conformal structure on  $\mathscr{M}$, whose coordinates are free to move in the range  
$-\pi/2\leq p\leq q\leq\pi/2$. The boundary is given by $p=-\pi/2$ or $q=\pi/2$ and the interior of 
$\mathscr{M}$ is conformally equivalent to Minkowski space-time.\\
We introduce the following points in $\mathscr{M}$:
\begin{itemize}
\item $i^+$, called \textit{future timelike infinity} given by the limits $t\pm r\rightarrow\infty$, 
$u,v\rightarrow\infty$, $p,q\rightarrow\frac{\pi}{2}$, $t'\rightarrow\pi$, $r'\rightarrow 0$. 
All the images in $\mathscr{M}$ of timelike geodesics terminate at this point;
\item $i^-$, called \textit{past timelike infinity} given by the limits $t\pm r\rightarrow -\infty$, 
$u,v\rightarrow-\infty$, $p,q\rightarrow-\frac{\pi}{2}$, $t'\rightarrow -\pi$, $r'\rightarrow 0$.  
All the images in $\mathscr{M}$ of timelike geodesics originate at this point;
\item $i^0$, called \textit{spacelike infinity} given by the limits $t\pm r\rightarrow\pm \infty$, 
$u\rightarrow-\infty$, $v\rightarrow\infty$, $p\rightarrow-\frac{\pi}{2}$, $q\rightarrow\frac{\pi}{2}$, 
$t'\rightarrow0$, $r'\rightarrow \pi$. All spacelike geodesics originate and terminate at this point.
\end{itemize}
We also introduce the following hypersurfaces in $\mathscr{M}$:
\begin{itemize}
\item $\mathscr{I^+}$, called \textit{future null infinity}, is the null hypersurface where all the 
outgoing null geodesics terminate and is obtained in the following way. Null outgoing geodesics are 
described by $t=r+c$, with $c$ finite constant, from which $u=t-r=c$ and $v=t+r=2t-c$. Taking the limit 
$t\rightarrow\infty$ we get $u=c$ and $v=\infty$, hence $q=\pi/2$ and $p=\tan^{-1}c=p_0$ with 
$-\pi/2< p_0<\pi/2$. In $(t',r')$ coordinates $t'=\pi/2+p_0$ and $r'=\pi/2-p_0$. As $p_0$ runs in its 
range of values this is a point moving on the segment connecting $i^+$ and $i^0$. All outgoing null 
geodesics terminate on this segment, described by the equation $t'=\pi-r'$.
\item $\mathscr{I^-}$, called \textit{past null infinity}, is the hypersurface form which all null 
ingoing geodesics originate. It can be shown that this is given by the region $p=-\pi/2$ and 
$-\pi/2<q_0<\pi/2$ and is described, in terms of $(t',r')$ coordinates, by the segment of equation 
$t'=\pi+r'$ connecting $i^{-}$ and $i^{0}$.
\end{itemize}
\begin{figure}[h]
\begin{center}
\includegraphics[scale=0.45]{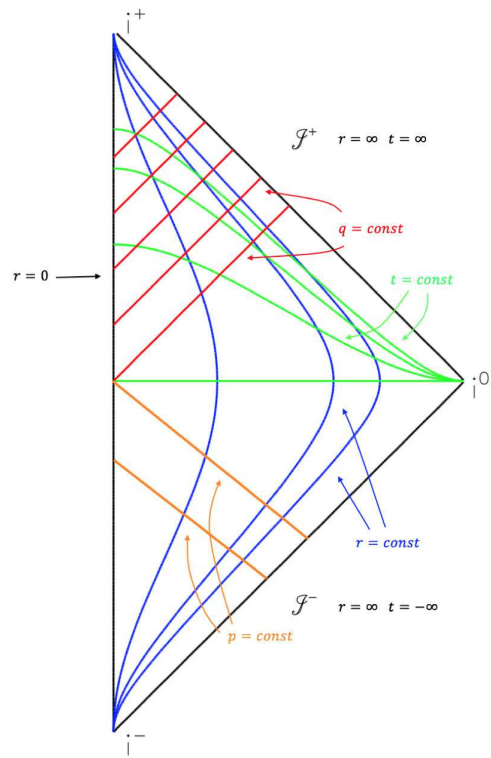}
\caption{A Penrose diagram for $\mathscr{M}$, using $(t',r')$ coordinates.} 
\label{fig:2.3.1}
\end{center}
\end{figure}
Putting 
\begin{equation*}
f^{\pm}(t',r')=t'\pm r'-\pi,
\end{equation*} 
the two equations defining the hypersurfaces $\mathscr{I}^+$ and $\mathscr{I}^-$ are 
\begin{equation*}
f^{\pm}(t',r')=0,
\end{equation*}
respectively. 
The normal co-vectors to $\mathscr{I}^+$ and $\mathscr{I}^-$ are 
\begin{equation*}
n^{\pm}_a=\frac{\partial f^{\pm}}{\partial x^a}=(1,\pm1,0,0).
\end{equation*}
Since $g^{ab}n^{\pm}_an^{\pm}_b=0$ it follows that $\mathscr{I}^+$ and $\mathscr{I}^-$ are null hypersurfaces.\\
At this stage, we can build some useful representation of the space-time $\mathscr{M}$. One of them is depicting 
$\mathscr{M}$ as a portion of the cylinder $\mathscr{E}=S^3\times E^1$, see Figure \ref{fig:2.3}. Another one is a portion of the plane in $(t',r')$ coordinates, that is an example of Penrose diagram. Each 
point of the Penrose diagram represents a sphere $S^2$, and radial null geodesics are represented by 
straight lines at $\pm 45^{\circ}$, see Figure \ref{fig:2.3.1}
\\One more representation for Minkowski space-time is furnished by Figure \ref{fig:2.4}.
\begin{figure}[h]
\begin{center}
\includegraphics[scale=0.43]{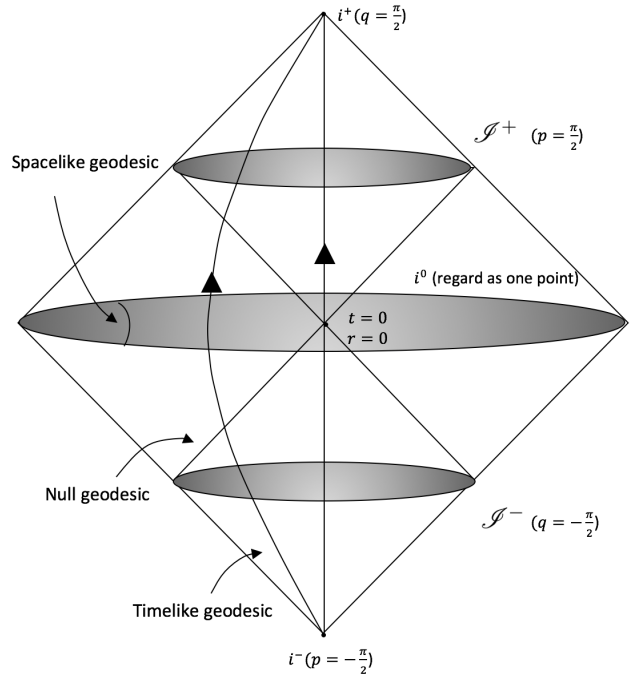}
\caption{This is another useful way of depicting $\mathscr{M}$ as the interior of two cones joined base to base. 
This picture however is not conformally accurate: in fact $i^0$ appears as an equatorial region 
whereas it should be a point.}
\label{fig:2.4}
\end{center}
\end{figure}
\\ \\We note here that for $\mathscr{M}$ the points $i^+$, $i^-$ and $i^0$ are regular and that $\mathscr{I^+}$ 
and $\mathscr{I^-}$ both have $S^2\times\mathbb{R}$ topology. Furthermore, the boundary of $\mathscr{M}$ 
is given by $\mathscr{I}=\mathscr{I^+}\cup\mathscr{I^-}\cup i^+\cup i^-\cup i^0$.
\\Consider now Schwarzschild space-time, with metric
\begin{equation}
\label{eqn:4}
\tilde{g}=dt\otimes dt\left(1-\frac{2m}{r}\right)-dr\otimes dr\left(1-\frac{2m}{r}\right)^{-1}-r^2\Sigma_2.
\end{equation}
Introducing $(u,w)$ coordinates as \begin{equation}
\label{eqn:5}
u=t-\left[r+2m\ln\left(\frac{r}{2m}-1\right)\right],\hspace{1cm}w=1/r,
\end{equation}
we have
\begin{equation}
\label{eqn:5.1}
\tilde{g}=du\otimes du\left(1-2mw\right)-(du\otimes dw+dw\otimes du)\frac{1}{w^2}-\frac{1}{w^2}\Sigma_2.
\end{equation}
The first of \eqref{eqn:5} is just the null retarded coordinate, corresponding to a null outgoing geodesic. 
Note that the coordinate $r^*=r+2m\ln\left(r/2m-1\right)$ in \eqref{eqn:5} is the usual Wheeler-Regge 
\lq tortoise coordinate\rq\hspace{0.1mm} introduced in \cite{Wheel}.
Consider now the unphysical metric
\begin{equation*}
g=\Omega^2d\tilde{g},\hspace{1cm}\Omega=w,
\end{equation*}
\begin{equation}
\label{eqn:6}
g=w^2(1-2mw)du\otimes du-(du\otimes dw+dw\otimes du)-\Sigma_2.
\end{equation}
Schwarzschild space-time, $\mathscr{\tilde M}$, is given by $0<w<1/2m$ because $2m<r<\infty$. We remark 
that the Schwarzschild solution can be easily extended beyond the event horizon, i.e. $0<r<\infty$ and 
$0<w<\infty$ because the apparent singular point $r=2m$ of the metric \eqref{eqn:4} is just a coordinate 
singularity and not a physical one, as can be noticed from \eqref{eqn:5.1}. The metric \eqref{eqn:6} is 
defined for $w=0$ (i.e. $r=\infty$) and hence for $\mathscr{M}$  we may take the range $0\leq w<1/2m$, 
such that the hypersurface $\mathscr{I^+}$ is given by $\Omega=w=0$. \\ Re-expressing \eqref{eqn:6} in 
terms of a null advanced coordinate
\begin{equation*}
v=u+2r+4m\ln\left(\frac{r}{2m}-1\right),
\end{equation*}
corresponding to a null ingoing geodesic we get
\begin{equation}
\label{eqn:7}
g=w^2(1-2mw)dv\otimes dv+(dv\otimes dw+dw\otimes dv)-\Sigma_2.
\end{equation}
By doing this it is now possible to introduce $\mathscr{I^-}$ as the hypersurface of $\mathscr{M}$ described 
by \eqref{eqn:7} for $w=0$. It is easy to check that the hypersurfaces $\mathscr{I}^+$ and $\mathscr{I}^-$, 
given by the equations $f^{\pm}(w)=w=0$ are again null hypersurfaces. 
\begin{figure}[h]
\begin{center}
\includegraphics[scale=0.51]{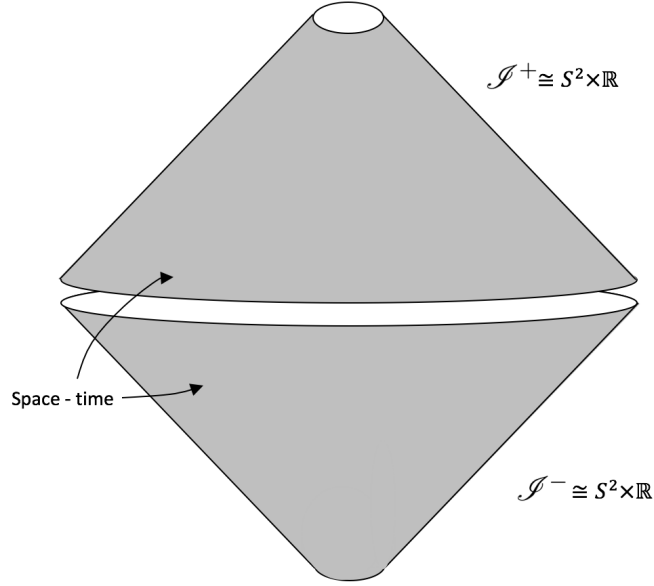}
\caption{Null infinity for Schwarzschild space-time. Note that $w=0$ corresponds both to $\mathscr{I^+}$ 
and $\mathscr{I^-}$. The points $i^{\pm}$ and $i^0$ are singular and have been deleted.}
\label{fig:2.5}
\end{center}
\end{figure}
\\The main difference between the Minkowski space-time case emerges from the fact that the points $i^+$, $i^-$ 
and $i^0$ in the Schwarzschild case are not regular, as could be deduced by the study of the eigenvalues of 
the Weyl tensor. However, it should not be surprising that $i^+$ and $i^-$ turn out to be singular, since the 
source generating the gravitational field becomes concentrated at these points, at the two ends of its history. 
Thus, we will omit $i^+$, $i^-$ and $i^0$ from the definition of $\mathscr{I}$, that will be just  
$\mathscr{I}=\mathscr{I^-}\cup\mathscr{I^+}$. We have two disjoint boundary null hypersurfaces $\mathscr{I^-}$ 
and $\mathscr{I^+}$ each of which is a cylinder with topology $S^2\times\mathbb{R}$. These null hypersurfaces 
are generated by rays (given by $\theta$,$\phi=$constant, $w=0$) whose tangents are normals to the hypersurfaces. 
These rays may be taken to be the $\mathbb{R}'\mathrm{s}$ of the topological product $S^2\times\mathbb{R}$. An useful representation of the Schwarzschild space-time is furnished by Figure \ref{fig:2.5}.\\
Take now into account a space-time $(\mathscr{\tilde M},\tilde g)$ with metric tensor \cite{Pen67,Penrin2}
\begin{equation}
\label{eqn:7.1}
\tilde g=r^{-2}Adr\otimes dr+B_i(dx^i\otimes dr+dr\otimes dx^i)+r^2C_{ij}dx^i\otimes dx^j,
\end{equation}
with $A$, $B_i$ and $C_{ij}$ sufficiently differentiable functions (say $C^3$) of $x^{\mu}$, with $x^0=r^{-1}$, 
on the hypersurface $\mathscr{I}$ defined by $x^0=0$ and in its neighbourhood. If the determinant 
\begin{equation*}
\mathrm{det}\left(\begin{matrix}A & B_i \\ B_j & C_{ij}\end{matrix}\right)
\end{equation*}
does not vanish, the space-time $(\mathscr{M},g)$ with metric $g=\Omega^2\tilde{g}$, being  $\Omega=r^{-1}$,
\begin{equation*}
g=Adx^0\otimes dx^0-B_i(dx^i\otimes dx^0+dx^0\otimes dx^i)+C_{ij}dx^i\otimes dx^j
\end{equation*}
is regular on $\mathscr{I}$. It is clear that Schwarzschild space-time is just a particular case of this 
more general situation described by \eqref{eqn:7.1}. Furthermore, this metric includes all
metrics of Bondi-Sachs type and describes a situation where there is an isolated source (with 
asymptotic flatness) and outgoing gravitational radiation. Hence a regularity assumption for $\mathscr{I}$ 
seems a not unreasonable one to impose if we wish to study asymptotically flat space-times and allow the 
possibility of gravitational radiation. In such situations, therefore, we expect a future-null conformal 
infinity $\mathscr{I}$ to exist. The choice made for $\Omega$ possesses the important property that its 
gradient at $\mathscr{I}$, $\partial\Omega/\partial x^{\mu}=(1,0,0,0)$, is not vanishing and hence defines 
a normal direction to $\mathscr{I}$ ($\mathscr{I}$ being described by the equation $\Omega=0$).\\
Roughly speaking, to say that a space-time is asymptotically flat means that its infinity is 
\lq similar\rq\hspace{0.1mm} in some way to the Minkowski one. As a consequence we may expect the 
conformal structure at infinity of an asymptotically flat space-time 
to be similar to the one found for the Minkowski case.\\
With those ideas in mind we may now proceed to a rigorous definition of asymptotic simplicity for a 
space-time. However we must also bear in mind that asymptotic flatness is, by itself, a mathematical idealization, and hence mathematical convenience and elegance constitute, by themselves, important criteria for selecting the appropriate idealization.
\begin{defn} $\\ $
\label{defn:AS}
A space-time $(\mathscr{\tilde M},\tilde g)$ is \textit{$k$-asymptotically simple} if some $C^{k+1}$ 
smooth manifold-with-boundary $\mathscr{M}$, with metric $g$ and smooth boundary 
$\mathscr{I}=\mathscr{\partial M}$ exists such that:
\begin{enumerate}
\item $\mathscr{\tilde M}$ is an open sub-manifold of $\mathscr{M}$; 
\item there exists a real-valued and positive function $\Omega>0$, that is $C^k$ throughout $\mathscr{M}$, 
such that $g_{ab}=\Omega^2\tilde g_{ab}$ on $\mathscr{\tilde{M}}$;
\item $\Omega=0$ and $\nabla_a\Omega\neq 0$ on $\mathscr{I}$;
\item every null geodesic on $\mathscr{M}$ has two endpoints on $\mathscr{I}$.
\end{enumerate}
The space-time $(\mathscr{\tilde M},\tilde g)$ is called \textit{physical space-time}, while 
$(\mathscr{M},g)$ is the \textit{unphysical space-time}.
\end{defn}
\begin{defn}{\cite{HawEll}}$\\ $
\label{defn:EAS}
A space-time $(\mathscr{\tilde M},\tilde g)$ is \textit{$k$-asymptotically empty and simple} if it is 
$k$-asymptotically simple and if satisfies the additional condition
\begin{enumerate}[start=5]
\item $\tilde{R}_{ab}=0$ on an open neighbourhood of $\mathscr{I}$ in $\mathscr{M}$ (this condition can 
be modified to allow for the existence of electromagnetic radiation near $\mathscr{I}$).
\end{enumerate}
\end{defn}
\begin{oss}$\\ $
Note that there are many different definitions of asymptotic simplicity. We used here the one which is 
due to \cite{Penrin2}, but others which slightly differ from this are conceivable \cite{Pen67,HawEll,Stew}.
\end{oss}
\begin{oss}$\\ $
Note that, although the extended manifold $\mathscr{M}$ and its metric are called \lq unphysical\rq, 
there is nothing unphysical in this construction. The boundary of $\mathscr{\tilde{M}}$ in $\mathscr{M}$ is 
uniquely determined by the conformal structure of $\mathscr{\tilde{M}}$ and, therefore, it is just 
as physical as $\mathscr{\tilde{M}}$. 
\end{oss}
Now we try to justify the previous assumptions.\\
Clearly with $1.$, $2.$ and $3.$ we mean to build $\mathscr{I}$ as the null infinity of 
$(\mathscr{\tilde M},\tilde{g})$, using the results obtained in the Minkowski case, with which it must share some 
properties. Condition $4.$ ensures that the whole of null infinity is included in $\mathscr{I}$. 
Furthermore, null geodesics in $\mathscr{\tilde M}$ correspond to null geodesics in $\mathscr{M}$ because 
conformal transformations map null vectors to null vectors: the concept of null geodesic is conformally invariant. 
Thus, we deduce that past and future infinity of any null geodesic in  
$\mathscr{\tilde M}$ are points of $\mathscr{I}$. Condition $5.$ ensures that the physical Ricci curvature 
$\tilde{R}_{ab}$ vanishes in the asymptotic region far away from the source of the gravitational field. 
Finally note how the points $i^+$, $i^-$ and $i^0$ are ruled out from the definition of $\mathscr{I}$, since 
$\mathscr{I}$ is not a smooth manifold at these points. 
Now we briefly summarize some of the properties of an asymptotically simple space-time, under the assumption that 
the vacuum Einstein equations hold and hence the cosmological constant equals zero.
\begin{itemize}
\item \textit{$\mathscr{I}$ is a null hypersurface}\\
This is because of condition $5.$ and condition $3.$. In fact it is easy to see that the Ricci scalar $R$ of 
the metric $g_{ab}$ is related to the Ricci scalar $\tilde{R}$ of the metric $\tilde{g}_{ab}$ by 
\begin{equation*}
\tilde{R}=\Omega^{-2}R -6\Omega^{-1}g^{cd}\nabla_c\nabla_d\Omega+3\Omega^{-2}g^{cd}\nabla_c\Omega\nabla_d\Omega,
\end{equation*}
and hence, by multiplying both members by $\Omega^2$, and by evaluating this equation on $\mathscr{I}$ where 
$\Omega=0$, it follows that $g^{cd}\nabla_c\Omega\nabla_d\Omega=0$. By condition $3.$, since 
$\nabla_c\Omega\neq 0$, it follows that 
$g^{cd}\nabla_c\Omega\nabla_d\Omega =0$ and thus $\nabla_c\Omega$, the normal vector to $\mathscr{I}$, is 
null and, by definition, $\mathscr{I}$ is a null hypersurface;\\
\item \textit{$\mathscr{I}$ is shear-free}\\
$R_{ab}$ is related to $\tilde{R}_{ab}$ by 
\begin{equation*}
\tilde{R}_{ab}=R_{ab}-2\Omega^{-1}\nabla_a\nabla_b\Omega-g_{ab}(\Omega^{-1}\nabla_c\nabla^c\Omega
-3\Omega^{-2}\nabla_c\Omega\nabla^c\Omega).
\end{equation*}
Since $\nabla_c\Omega$ is null and $R_{ab}$ is defined on $	\mathscr{I}^+$, if condition $5.$ holds, 
the previous equation on $\mathscr{I}^+$ leads to
\begin{equation*}
2\nabla_a\nabla_b\Omega+g_{ab}\nabla_c\nabla^c\Omega=0. 
\end{equation*}
Contracting with $g^{ab}$ it gives 
\begin{equation*}
\nabla_c\nabla^c\Omega=0\Rightarrow\nabla_a\nabla_b\Omega=0.
\end{equation*}
Hence the normal vector to $\mathscr{I}$ is divergence- and shear-free;\\
\item \textit{$\mathscr{I}^+$ has two connected components, $\mathscr{I}^+$ and $\mathscr{I}^-$, 
each of which has topology $S^2\times\mathbb{R}$}\\
The first proof of this theorem, involving sophisticated arguments, is due to \cite{Pen65}. However, 
as remarked in \cite{NewRP}, the arguments carried out by Penrose are incorrect, and a more rigorous 
proof can be found in \cite{Ger71} or in \cite{HawEll}. The significance of this property lies in the 
fact that the structure of the conformal infinity found for Minkowski space-time is that 
of any asymptotically simple space-time.
\end{itemize}
At this stage we must make a clarification. In fact we must point out that condition $4.$ is difficult to verify in practice and is not even satisfied by some space-times 
that we would like to classify as asymptotically flat. As an example, for Schwarzschild space-time, 
it is known that there exist null circular orbits with radius $3m$, and hence do not terminate on 
$\mathscr{I^+}$. For these reasons condition $4.$ is often too strong and gets replaced by a weaker 
one that brings to the notion of weakly asymptotically simple space-time. 
\begin{defn} $\\ $
\label{defn:WAS}
A space-time ($\mathscr{\tilde M},\tilde g)$ is \textit{weakly asymptotically simple} if there exists an 
asymptotically simple space-time $(\mathscr{\tilde M'},\tilde g')$ with associated unphysical space-time 
$(\mathscr{M'},g')$, such that for a neighbourhood  $\mathscr{H'}$ of $\mathscr{I'}$ in $\mathscr{M'}$, 
the region $\mathscr{\tilde M'}\cap\mathscr{H'}$ is isometric to a similar neighbourhood 
$\mathscr{\tilde H}$ of $\mathscr{\tilde M}$.
\end{defn}
In this way a weakly asymptotically simple space-time possesses the same properties of the conformal 
infinity of an asymptotically simple one, but the null geodesics do not necessary reach it because it 
may have other infinities as well. Such space-times are essentially required to be isometric to an 
asymptotically simple space-time in a neighbourhood of $\mathscr{I}$. 
\begin{oss}$\\ $
Note that the definition \ref{defn:AF1} of asymptotic flatness seems to be completely different from 
that of asymptotic simplicity \ref{defn:AS} and weak asymptotic simplcity \ref{defn:EAS}. However, 
the two approaches are equivalent, as shown in \cite{New62,NewUn}, since they lead to the same asymptotic 
properties, using two different ways. It is worth noting that the conformal method introduced by 
Penrose represents a \lq natural evolution\rq\hspace{0.1mm} of the previous one, being more geometrical.
\end{oss}
\section{Symmetries on $\mathscr{I}$}
\label{sect:5}
The geometrical approach to asymptotic flatness, discussed in the previous section, affords us a 
much more vivid picture of the significance of the BMS group.\\
The idea is that, by adjoining to the physical space-time $(\tilde{\mathscr{M}},\tilde{g})$ an appropriate 
conformal boundary $\mathscr{I}$, as done in Sect. \ref{sect:4}, we may obtain the asymptotic symmetries as 
conformal transformations of the boundary, the boundary having a much better chance of having a meaningful 
symmetry group than $\mathscr{\tilde{M}}$. \\ We start by making an example to better understand the nature 
of the problem, which is due to \cite{Pen72}. Consider Minkowski space-time with standard coordinates 
$(t,x,y,z)$, the metric being given by
\begin{equation*}
g=\eta_{ab}dx^a\otimes dx^b=dt\otimes dt-dx\otimes dx-dy\otimes dy-dz\otimes dz,
\end{equation*} 
and consider the null cone $\mathscr{N}$ through the origin, given by the equation 
\begin{equation}
\label{eqn:162}
t^2-x^2-y^2-z^2=0.
\end{equation}
The generators of $\mathscr{N}$ are the null rays through the origin, given by
\begin{equation*}
t:x:y:z=\mathrm{const},
\end{equation*}
with $t,x,y,z$ satisfying \eqref{eqn:162}. Let us consider $S^2$ to be the section of $\mathscr{N}$ by 
the spacelike 3-plane $t=1$. Then there exists a (1-1)-correspondence between the generators of $\mathscr{N}$ 
and the points of $S^2$ (i.e. that given by the intersections of the generators with $t=1$). We may regard 
$S^2$ as a realization of the space of generators of $\mathscr{N}$. However, we could have used any 
other cross-section $\hat{S}^2$ of $\mathscr{N}$ to represent this space. The important point is to 
realize that the map which carries any one such cross-section into another, with points on the same 
generator of $\mathscr{N}$ corresponding to one another, is a conformal map. The situation is 
reported in Figure \ref{fig:6.1}.:
\begin{figure}[h]
\begin{center}
\includegraphics[scale=0.3]{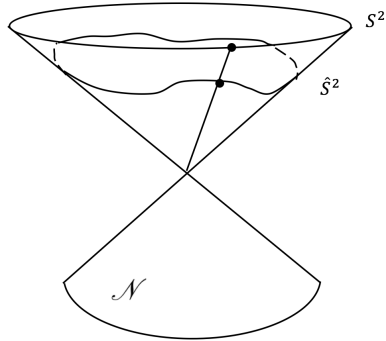}
\caption{The generators of the null cone $\mathscr{N}$ establish a 1-1 map between any two 
cross-sections of $\mathscr{N}$.}
\label{fig:6.1}
\end{center}
\end{figure}
\\The above mentioned map being conformal, the space of generators of $\mathscr{N}$ may itself be 
assigned a conformal structure, i.e. that of any of these sections. To see that the map is conformal 
we may re-express the induced metric on $\mathscr{N}$ in the form
\begin{equation}
\label{eqn:163}
g_{_{\mathscr{N}}}=-r^2\gamma_{\alpha\beta}(x^{\gamma})dx^{\alpha}\otimes dx^{\beta}+0\cdot dr\otimes dr,
\end{equation}
where $x^{\alpha}$ and $r$ are coordinates on $\mathscr{N}$, the generators being given by the coordinate 
lines $x^{\alpha}=\mathrm{const}$ (the term \lq$0$\rq\hspace{0.1mm} takes into account that, the surface 
$\mathscr{N}$ being null, its induced metric is degenerate, i.e. with vanishing determinant). There exist 
obviously many ways of attaining the form \eqref{eqn:163}. One is to use ordinary spherical coordinates 
for Minkowski space-time, giving $g_{_{\mathscr{N}}}=-r^2(d\theta\otimes d\theta+\sin^2\theta d\phi\otimes d\phi)
+0\cdot dr\otimes dr$. Since a cross-section of $\mathscr{N}$ is given by specifying $r$ as function of 
$x^{\alpha}$ it is clear that any two cross-sections give conformally related metrics, being mapped to one 
another by the generators of $\mathscr{N}$. It is now obvious that many other cone-like null surfaces will 
share this property of $\mathscr{N}$, provided their metrics can be put in the form \eqref{eqn:163}. 
Now if we suppose here to deal with an empty asymptotically simple space-time $(\mathscr{\tilde{M}},\tilde{g})$ 
(according to definition \ref{defn:EAS}, with associated unphysical space-time $(\mathscr{M},g)$) 
we know that, if $\mathscr{I}$ is null,  it has the important property to be shear-free, as discussed in 
Sect. \ref{sect:4}. Physically, the shear-free nature of the generators of $\mathscr{I}$ tells us that 
small shapes are preserved as we follow these generators along $\mathscr{I}$. Hence any diffeomorphism 
which maps each null generator of $\mathscr{I}^+$ into itself is a conformal transformation for any metric 
on $\mathscr{I}^+$. That is to say, if we take any two cross-sections $S_1$ and $S_2$ of $\mathscr{I}^+$ 
or $\mathscr{I}^-$, then the correspondence between $S_1$ and $S_2$ established by the generators is a 
conformal one. This is exactly the same situation we encountered in the example with $\mathscr{N}$. 
We have the following
\begin{prop}$\\ $
\label{prop:gen}
If $\mathscr{I}$ is null, then any two cross-sections of $\mathscr{I}^{\pm}$ are mapped to one another 
conformally by the generators of $\mathscr{I}^{\pm}$. 
\end{prop}
In Sect. \ref{sect:4} we have discussed that the topology of $\mathscr{I}^{\pm}$ is $S^2\times\mathbb{R}$, 
where the $\mathbb{R}$ factor may be taken as the null-geodesic generator $\mathscr{I}^{\pm}$. Hence these 
generators, by proposition \ref{prop:gen}, establish a conformal mapping between any two $S^2$ cross-sections 
of $\mathscr{I}^{\pm}$, these sections being of course conformal spheres. It is a theorem that any conformal 
2-surface with the topology of a sphere $S^2$ is conformal to the unit 2-sphere in Euclidean 3-space. 
Thus we can assume without loss of generality, that the conformal factor $\Omega$ has been chosen so that 
some cross-section $S$ has unphysical squared line element $-ds^2$ of a unit 2-sphere. Given one choice of 
$\Omega$, we can always make a new choice $\Omega'=\Theta\Omega$ which again has the property of vanishing at 
$\mathscr{I}$ with non-zero gradient there. The factor $\Theta$ has to be an arbitrary smooth positive function 
on $\mathscr{I}$ and can be chosen to rescale the metric on $\mathscr{I}$ as we please. It is worth noting 
that the shear-free condition can be saved by the change $\Omega'=\Theta\Omega$, as discussed in \cite{Stew}. 
This property can be interpreted as a \lq gauge freedom\rq\hspace{0.1mm} in the choice of the conformal 
factor $\Omega$. We can use this freedom to set the metric of a continuous sequence of cross-sections 
along the generators equal to that of $S$. Hence, in spherical polar coordinates the induced metric 
on $\mathscr{I}^{+}$ is 
\begin{equation}
\label{eqn:164}
g_{_{\mathscr{I^+}}}=d\theta\otimes d\theta+\sin^2\theta d\phi\otimes d\phi+0\cdot du\otimes du,
\end{equation}
where $u$ is a retarded time coordinate, i.e. a parameter defined along each generator increasing monotonically 
with time from $-\infty$ to $+\infty$, the corresponding form with an advanced time coordinate $v$ in place 
of $u$ holding for $\mathscr{I}^{-}$. The surfaces $u=\mathrm{const}$ are cross-sections of $\mathscr{I}^+$, 
each of which has the metric of a unit 2-sphere, as is clear from \eqref{eqn:164}. \\
From the above discussion it follows that the metric on $\mathscr{I}^+$ belongs to an equivalence class of 
metrics, two elements being equivalent if they are conformally related one to the other. Hence, the form 
of the metric \eqref{eqn:164} is just one element of this equivalence class that we have chosen as 
representative. Let us consider the group of conformal transformations of $\mathscr{I}^{+}$, i.e. the 
group of transformations which conformally preserve the metric \eqref{eqn:164}. It is clear that any smooth 
transformation which maps each generator into itself will be allowable:
\begin{equation}
\label{eqn:165}
u\rightarrow u'=F(u,\theta,\phi),
\end{equation}
with $F$ smooth on the whole $\mathscr{I}^+$ and $\partial F/\partial u >0$, since it has to map the whole 
range for $u$ to itself, for any $\theta$ and $\phi$. In addition, we can allow conformal transformations 
of the $(\theta,\phi)$-sphere into itself. These transformations can be regarded as those of the 
compactified complex plane $\mathbb{C}\cup\{\zeta=\infty\}$ into itself. 
Introducing the complex stereographic coordinate
\begin{equation*}
\zeta=e^{i\phi}\cot\frac{\theta}{2},
\end{equation*}
we have that \eqref{eqn:164} may be written as 
\begin{equation}
\label{eqn:167}
g_{_{\mathscr{I^+}}}=\frac{2(d\zeta \otimes d\bar{\zeta}+d\bar{\zeta}\otimes d\zeta)}
{(1+\zeta\bar{\zeta})^2}+0\cdot du\otimes du.
\end{equation}
Then the most general conformal transformation of the compactified plane is given by
\begin{equation}
\label{eqn:166}
\zeta\rightarrow\zeta'=\frac{a\zeta+b}{c\zeta+d},
\end{equation}
where $a$,$b$,$c$,$d\in\mathbb{C}$, that can be normalized to satisfy $ad-bc=1$.
\begin{oss}$\\ $
Since conformal transformations can be equivalently expressed in terms of $x^A$ or $\zeta$ coordinates, 
in the remainder we will use both of them, depending on the convenience.
\end{oss}
The particular functional form of the transformations in \eqref{eqn:166} results from the request that 
they must be diffeomorphisms of the compactified plane $\mathbb{C}\cup\{\zeta=\infty\}$  into itself. 
Hence the transformations must have at least one pole, at $\zeta^*$ say, corresponding to the point 
that is mapped to the north pole $F(\zeta^*)=\infty$ and at least one zero, at $\zeta^{**}$ say, 
corresponding to the point that is mapped to the south pole $F(\zeta^{**})=0$. Thus, the transformations 
must be some rational complex function where the roots of the numerator and the denominator correspond 
to the points that are mapped to the south and the north pole, respectively. Since the transformation 
must be injective there must be one, and only one, point that is mapped to the south pole, and also 
exactly one other point that is mapped to the north pole. This requires that both numerator and denominator 
be linear functions of $\zeta$. Requiring this map to be surjective finally imposes that the complex 
numbers $a,b,c,d$ in \eqref{eqn:166} must satisfy $ad-bc\neq0$ (all of these parameters can be appropriately 
rescaled to get $ad-bc=1$ leaving the transformation unchanged). It is worth remarking that in pure 
mathematics these transformations were studied by Poincar\'e and other authors when they developed the 
theory of what are nowadays called automorphic functions, i.e. meromorphic functions such that 
$f(z)=f((az+b)/(cz+d))$ \cite{PO1912}. It is easy to see that these transformations contain:
\begin{itemize}
\item Translations $\zeta\rightarrow \zeta'=\zeta+b,\hspace{1cm}b\in\mathbb{C}$;
\item Rotations $\zeta\rightarrow \zeta'=e^{i\theta}\zeta,\hspace{1cm}\theta\in\mathbb{R}$;
\item Dilations $\zeta\rightarrow \zeta'=e^{-\chi}\zeta,\hspace{1cm}\chi\in\mathbb{R}$;
\item Special transformations $\zeta\rightarrow \zeta'=
-\displaystyle{\frac{b^2}{\zeta^2}},\hspace{1cm}b\in\mathbb{C}$;
\end{itemize}
Any transformation of the form \eqref{eqn:166} can be obtained as the composition of a special transformation, 
a translation, a rotation and a dilation.\\
Usually, transformations \eqref{eqn:166} are referred to as the \textit{conformal group} (in two dimensions), 
the \textit{projective linear group}, the \textit{M\"{o}bius transformations} or the \textit{fractional linear 
transformations}, and is denoted by PSL$(2,\mathbb{C})\cong\mathrm{SL}(2,\mathbb{C})/\mathbb{Z}_2$ 
(as will be discussed in the next section). Under these transformations we have
\begin{equation*}
\frac{2(d\zeta'\otimes d\bar{\zeta'}+d\bar{\zeta'}\otimes d\zeta')}{(1+\zeta'\bar{\zeta'})^2}
=K^2(\zeta,\bar{\zeta})\frac{2(d\zeta\otimes d\bar{\zeta}+d\bar{\zeta}\otimes d\zeta)}
{(1+\zeta\bar{\zeta})^2}\Rightarrow g'_{_{\mathscr{I}^+}}=K^2g_{_{\mathscr{I}^+}},
\end{equation*}
with 
\begin{equation}
\label{eqn:170}
K(\zeta,\bar{\zeta})=\frac{1+\zeta\bar{\zeta}}{(a\zeta+b)(\bar{a}\bar{\zeta}+\bar{b})
+(c\zeta+d)(\bar{c}\bar{\zeta}+\bar{d})}.
\end{equation}
It can be shown that transformations \eqref{eqn:166} are equivalent to \eqref{eqn:182} and \eqref{eqn:183}, 
and that the conformal factor $K$ in \eqref{eqn:170} is the same as one in \eqref{eqn:171}, 
expressed in terms of the $(\theta,\phi)$ variables.\\
\begin{defn}$\\ $
The group of transformations
\begin{subequations}
\begin{align}
\label{eqn:172}
&\zeta\rightarrow\zeta'=\frac{a\zeta+b}{c\zeta+d},\\
\label{eqn:172a}
&u\rightarrow u'=F(u,\theta,\phi),
\end{align}
\end{subequations}
with $ad-bc=1$ and with $F$ smooth and $\partial F/\partial u>0$ is the \textit{Newman-Unti (NU) group}.
\end{defn}
\begin{oss}$\\ $
Note that \eqref{eqn:172} are the non-reflective conformal transformations of the $S^2$-space of 
generators of $\mathscr{I}^+$ (the conformal structure being defined equivalently by one of its cross-sections), 
while \eqref{eqn:172a}, when \eqref{eqn:172} is the identity ($a=d=1$, $b=c=0$), give the general non-reflective 
smooth transformations of the generators to  themselves.
\end{oss}
The conformal metric \eqref{eqn:167} is considered to be part of the universal intrinsic structure of 
$\mathscr{I}^+$ (universal, in the sense that any space-time which is asymptotically simple and vacuum near 
$\mathscr{I}$ has a $\mathscr{I}^+$ metric and similarly a $\mathscr{I}^-$ metric which is conformal to 
\eqref{eqn:167}). Hence the NU group can be regarded as the group of non-reflective transformations of 
$\mathscr{I}^+$ preserving its intrinsic (degenerate) conformal metric \cite{Pen72,Pen82}.\\
However, the NU group is different from the BMS group, the former being larger than the latter. In fact the NU 
group allows a greater freedom in the function $F$, while in the BMS group $F$ is constrained to be of the 
form \eqref{eqn:184}. Thus, we want to be somehow able to reduce this freedom, assigning a further geometric 
structure to $\mathscr{I}^+$, the preservation of which will furnish the BMS group, restricting exactly the form of $F$ to be the one of 
\eqref{eqn:184}. This additional structure is referred to as the \textit{strong conformal geometry} 
\cite{Penrin2,Pen72}. The most direct way to specify this structure is the following. Consider a 
replacement of the conformal factor,
\begin{equation}
\label{eqn:173}
\Omega\rightarrow\Omega'=\Theta\Omega.
\end{equation}
We choose the function $\Theta$ to be smooth and positive on $\mathscr{M}$ and nowhere vanishing on 
$\mathscr{I}^+$. Under \eqref{eqn:173} the metric transforms as 
\begin{equation*}
g_{ab}\rightarrow g'_{ab}=\Theta^2g_{ab},\hspace{1cm}g^{ab}\rightarrow g'^{ab}=\Theta^{-2}g^{ab},
\end{equation*}
and the normal co-vector to $\mathscr{I}^{+}$ as
\begin{equation*}
N_a=-\nabla_a\Omega\rightarrow N'_a=-\nabla'_a\Omega'=-\nabla_a\Omega'=-\Omega\nabla_a\Theta
-\Theta\nabla_a\Omega\approx\Theta N_a,
\end{equation*}
while the vector
\begin{equation*}
N^a=g^{ab}\partial_b\Omega\rightarrow N'^a=g'^{ab}N'_b\approx\Theta^{-1}N^a,
\end{equation*}
where we introduced the \lq weak equality\rq\hspace{0.1mm} symbol $\approx$. Considering two fields 
$\psi^{...}_{...}$ and $\phi^{...}_{...}$, saying that 
\begin{equation}
\label{eqn:80}
\psi^{...}_{...}\approx\phi^{...}_{...}
\end{equation}
means that $\psi^{...}_{...}-\phi^{...}_{...}=0$ on $\mathscr{I}$. 
The line element $dl$ of $\mathscr{I}^+$ rescales according to 
\begin{equation}
\label{eqn:176}
dl\rightarrow dl'=\Theta dl.
\end{equation}
Having done any allowable choice of the conformal factor $\Omega$, through the function $\Theta$, and hence 
some specific choice of the metric $dl$ for cross-sections of $\mathscr{I}^+$, then it is defined, from 
$N_a=-\nabla_a\Omega$, a precise scaling for parameters $u$ on the generators of $\mathscr{I}^+$, fixed by
\begin{equation*}
\frac{\partial}{\partial u}=N^a\nabla_a,\hspace{0.7cm}\mathrm{i.e.}\hspace{0.7cm}N^a\nabla_a u=1.
\end{equation*}
Under \eqref{eqn:173} we see that to keep the scaling of the parameters $u$ along the generators fixed we must choose
\begin{equation}
\label{eqn:175}
du\rightarrow du'=\Theta du, 
\end{equation}
so that 
\begin{equation*}
N^a\nabla_a u\rightarrow N'^a\nabla'_au'=N'^a\nabla_au'=N'^a\frac{\partial u'}{\partial x^a}
=\Theta^{-1}\Theta N^a\nabla_a u=1.
\end{equation*}
All the parameters $u$, linked by \eqref{eqn:175}, scale in the same way along the generators of $\mathscr{I}^+$. 
From \eqref{eqn:176} and \eqref{eqn:175} we see that the ratio
\begin{equation}
\label{eqn:177}
dl:du
\end{equation}
remains invariant and it is independent of the choice of the conformal factor $\Omega$. It is the invariant 
structure provided by \eqref{eqn:177} that can be taken to define the strong conformal geometry. To better 
reformulate this invariance we introduce the concept of \textit{null angle} \cite{Pen63,Penrin2,Pen72,BMS75}. 
Consider two non-null tangent directions at a point $P$ of $\mathscr{I}^+$. Let $[X]$ and $[Y]$ be such 
directions. If no linear combination of $X\in [X]$ and $Y\in [Y]$ is the null tangent direction at $P$, 
then the angle between $[X]$ and $[Y]$ is defined by the metric \eqref{eqn:164}. However, if the null tangent 
direction at $P$ is contained in the plane spanned by $[X]$ and $[Y]$, then the angle between $[X]$ and $[Y]$ 
always vanishes. To see this choose $X\in [X]$, $Y\in [Y]$ and $N\in [N]$ ($[N]$ being the null direction 
tangent at $P$), such that $Y=X+N$. Then since $N$ is null we have, using the metric $g$ on $\mathscr{I}^+$ 
given in \eqref{eqn:164} (and hence any other one of its equivalence class)
\begin{equation*}
0=g(N,X)=g(Y-X,X)=g(Y,X)-g(X,X),
\end{equation*}
and 
\begin{equation*}
0=g(N,Y)=g(Y-X,Y)=g(Y,Y)-g(X,Y),
\end{equation*}
from which the angle $\theta$ between $[X]$ and $[Y]$, given by 
\begin{equation*}
\cos\theta=\frac{g(X,Y)}{\sqrt{g(X,X)g(Y,Y)}}=1,
\end{equation*}
vanishes. However, if we require the strong conformal geometry structure to hold and hence the invariance of 
the ratio \eqref{eqn:177} we can numerically define the null angle $\nu$ between two tangent directions at 
a point $P$ of $\mathscr{I}^{+}$ by 
\begin{equation}
\label{eqn:178}
\nu=\frac{\delta u}{\delta l},
\end{equation}
where the infinitesimal increments $\delta u$ and $\delta l$ are as indicated in Figure \ref{fig:6.3} \cite{Penrin2}.
\begin{figure}[h]
\begin{center}
\includegraphics[scale=0.4]{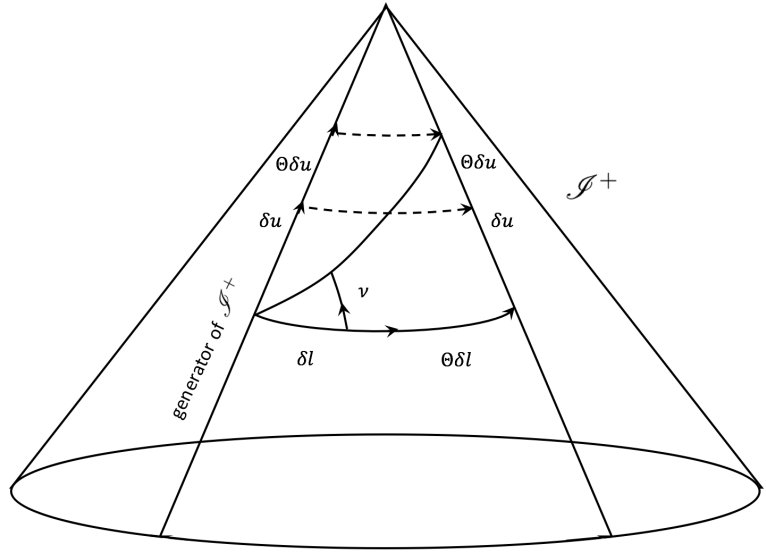}
\caption{A null angle $\nu$ on $\mathscr{I}^+$, given by $\nu=\delta u/\delta l$, is defined between a pair 
of directions on $\mathscr{I}^+$ whose span contains the null normal direction to $\mathscr{I}^+$.}
\label{fig:6.3}
\end{center}
\end{figure}
\\By virtue of the strong conformal geometry, under change of the conformal factor for the metric of $\mathscr{I}^+$, null angles 
remain invariant. For further insights about the strong conformal geometry and the interpretation of null angles we suggest to read \cite{Pen72} or \cite{BMS75}. \\A transformation of $\mathscr{I}^+$ to itself which 
preserves angles and null angles, i.e. that respects the strong conformal geometry structure, must have 
the effect that any expansion (or contraction) of the spatial distances $dl$ is accompanied by an equal 
expansion (or contraction) of the scaling of the special $u$ parameters. The allowed transformations
have the form \eqref{eqn:172}, where function $F$ must now have the  precise form that allows the ratio 
$du:dl$ to remain invariant. Under the transformation \eqref{eqn:173} we have, as seen, that the sphere 
of the cross-section of $\mathscr{I}^+$ undergoes a conformal mapping, i.e.
\begin{equation*}
dl\rightarrow dl'=\Theta dl.
\end{equation*}
Since $\Theta$ is the conformal factor of the transformation, it depends only on $\theta$ and $\phi$ or, 
equivalently, on $\zeta$ and $\bar{\zeta}$ and must have the form given in \eqref{eqn:170}. 
We must therefore also have 
\begin{equation*}
du\rightarrow du'=\Theta du, 
\end{equation*}
Integrating we get
\begin{equation*}
u\rightarrow u'=\Theta[u+\alpha(\zeta,\bar{\zeta})],
\end{equation*}
where $\Theta$ assumes the form
\begin{equation*}
\Theta(\zeta,\bar{\zeta})=\frac{1+\zeta\bar{\zeta}}{(a\zeta+b)(\bar{a}\bar{\zeta}+\bar{b})
+(c\zeta+d)(\bar{c}\bar{\zeta}+\bar{d})},
\end{equation*}
with $a,b,c,d\in\mathbb{C}$ and $ad-bc=1$. By virtue of definition \ref{defn:BMS} we have obtained the following
\begin{prop}$\\ $
The group of conformal transformations of $\mathscr{I}^+$ which preserve the strong conformal geometry, i.e. 
both angles and null angles, is the BMS group.
\end{prop}
\begin{oss}$\\ $
Conformal transformations, and hence the NU group, always preserve finite angles, but null angles, i.e. 
angles between tangent vectors of which $N^a$ is a linear combination, are preserved by the BMS group only. 
\end{oss}
The general form of a BMS transformation is thus
\begin{subequations}
\begin{equation}
\label{eqn:200}
\zeta\rightarrow\zeta'=\frac{a\zeta+b}{c\zeta+d},
\end{equation}
\begin{equation}
\label{eqn:201}
u\rightarrow u'=\frac{(1+\zeta\bar{\zeta})[u+\alpha(\zeta,\bar{\zeta})]}{(a\zeta+b)(\bar{a}\bar{\zeta}
+\bar{b})+(c\zeta+d)(\bar{c}\bar{\zeta}+\bar{d})},
\end{equation}
\end{subequations}
\\with $a,b,c,d\in\mathbb{C}$ and $ad-bc=1$. Clearly the BMS group is a subgroup of the NU group, the 
function $F$ having its form fixed. However it is still an infinite-dimensional function-space group.
\section{Structure of the BMS group}
\label{sect:6}
We discuss first the BMS transformations obtained by setting $\alpha=0$,
\begin{equation}
\label{eqn:220}
u\rightarrow u'=Ku,\hspace{1cm}\zeta\rightarrow\zeta'=\frac{a\zeta+b}{c\zeta+d},
\end{equation}
with 
\begin{equation*}
K(\zeta,\bar{\zeta})=\frac{1+\zeta\bar{\zeta}}{(a\zeta+b)(\bar{a}\bar{\zeta}+\bar{b})
+(c\zeta+d)(\bar{c}\bar{\zeta}+\bar{d})},
\end{equation*}
\\i.e. a rescaling for $u$ and a transformation of the group $\mathrm{PSL}(2,\mathbb{C})$ for $\zeta$. 
Any of these transformations is specified by the 4 constants $a,b,c,d\in\mathbb{C}$, satisfying $ad-bc=1$. 
Hence there are only 3 independent complex parameters, i.e. 6 independent real parameters. 
Any element $f\in \mathrm{PSL}(2,\mathbb{C})$ reads as
\begin{equation*}
f=\frac{a\zeta+b}{c\zeta+d}\equiv\{a,b,c,d\}.
\end{equation*}
Note that under simultaneous change $a\rightarrow-a$, $b\rightarrow-b$, $c\rightarrow-c$, 
$d\rightarrow-d$ any element $f\in\mathrm{PSL}(2,\mathbb{C})$ remains unaffected, i.e.
\begin{equation}
\label{eqn:203}
f=\{a,b,c,d\}=\{-a,-b,-c,-d\}.
\end{equation}
Now take into account the group $\mathrm{SL}(2,\mathbb{C})$ of $(2\times 2)$ complex matrices 
$Q$ with $\mathrm{det}[Q]=1$:
\begin{equation*}
Q=\left(\begin{matrix}A& B\\ C& D
\end{matrix}
\right),\hspace{1cm}\mathrm{det}[Q]=AD-BC=1,\hspace{1cm}A,B,C,D\in\mathbb{C}.
\end{equation*} 
Clearly the dimension of $\mathrm{SL}(2,\mathbb{C})$ is 6. Hence we can consider a map $\tilde{\varphi}$, 
between $\mathrm{PSL}(2,\mathbb{C})$ and $\mathrm{SL}(2,\mathbb{C})$ defined by
\begin{equation}
\label{eqn:202}
\tilde{\varphi}:f=\{a,b,c,d\}\in \mathrm{PSL}(2,\mathbb{C})\longrightarrow\tilde{\varphi}(f)
=\left(\begin{matrix}a & b\\
c & d\end{matrix}\right)\in\mathrm{SL}(2,\mathbb{C}).
\end{equation}
It is easy to show that, since the group operation of $\mathrm{PSL}(2,\mathbb{C})$ is the function 
composition $\circ$, given $f=\{a,b,c,d\}$ and $g=\{a',b',c',d'\}\in\mathrm{PSL}(2,\mathbb{C})$ we have
\begin{equation*}
f\circ g=\{aa'+bc',ab'+bd',ca'+dc',cb'+dd'\}\in\mathrm{PSL}(2,\mathbb{C}).
\end{equation*}
Then taking the images of $f$ and $g$ through $\tilde{\varphi}$,
\begin{equation*}
\tilde{\varphi}(f)=\left(\begin{matrix}a & b\\
c & d\end{matrix}\right),\hspace{1cm}\tilde{\varphi}(g)=\left(\begin{matrix}a' & b'\\
c' & d'\end{matrix}\right),
\end{equation*}
we have, since the operation in $\mathrm{SL}(2,\mathbb{C})$ is the ordinary matrix product,
\begin{equation}
\label{eqn:204}
\tilde{\varphi}(f)\cdot\tilde{\varphi}(g)=\left(\begin{matrix}aa'+bc' &ab'+bd'\\
ca'+dc' & cb'+dd'\end{matrix}\right)=\tilde{\varphi}(g\circ f).
\end{equation}
Note that, by virtue of \eqref{eqn:203}, to the same element $f$ there correspond, through $\tilde{\varphi}$, 
two different elements, $\tilde{\varphi}(f)$ and $-\tilde{\varphi}(f)$. If we consider now the map $\varphi$:
\begin{equation*}
\varphi:f=\{a,b,c,d\}\in \mathrm{PSL}(2,\mathbb{C})\longrightarrow\varphi(f)=\left(\begin{matrix}a & b\\
c & d\end{matrix}\right)\in\mathrm{SL}(2,\mathbb{C})/\mathbb{Z}_2.
\end{equation*}
it is clear that property \eqref{eqn:204} holds for $\varphi$ as well. This map is a group isomorphism, 
$\varphi(f)$ and $-\varphi(f)$ being now identified in $\mathrm{SL}(2,\mathbb{C})/\mathbb{Z}_2$. We can state
\begin{equation}
\label{eqn:205}
\mathrm{PSL}(2,\mathbb{C})\cong\mathrm{SL}(2,\mathbb{C})/\mathbb{Z}_2.
\end{equation}
The group $\mathrm{SL}(2,\mathbb{C})$ is the double covering of $\mathrm{PSL}(2,\mathbb{C})$. 
Furthermore, it is a well known result that
\begin{equation}
\label{eqn:210}
\mathscr{L}\cong\mathrm{SL}(2,\mathbb{C})/\mathbb{Z}_2,
\end{equation}
where $\mathscr{L}$ is the connected component of the Lorentz group. Thus
\begin{equation}
\label{eqn:211}
\mathrm{PSL}(2,\mathbb{C})\cong \mathscr{L}.
\end{equation}
We have the following
\begin{prop}$\\ $
The connected component of the Lorentz group is isomorphic with the subgroup $\mathrm{PSL}(2,\mathbb{C})$ of the BMS group.
\end{prop}
To make this isomorphism explicit take an element $f=\{a,b,c,d\}\in\mathrm{PSL}(2,\mathbb{C})$ and through 
$\varphi$ assign to it an element of $\mathrm{SL}(2,\mathbb{C})$,
\begin{equation*}
\varphi(f)=\left(\begin{matrix}a & b\\ c & d
\end{matrix}\right)\in\mathrm{SL}(2,\mathbb{C}),\hspace{1cm}ad-bc=1.
\end{equation*}
Then, using the isomorphism of \eqref{eqn:210} it can be shown with lengthy calculations \cite{Oblak16} 
that to $\varphi(f)$ there corresponds an element of $\mathscr{L}$ given by the matrix
\begin{equation}
\label{eqn:iso}
\Lambda(f)=
\end{equation}
\begin{equation*}
\left(\begin{matrix}\frac{1}{2}\left(|a|^2+|b|^2+|c|^2+|d|^2\right) & -\mathrm{Re}\left\{a\bar{b}+c\bar{d}\right\} 
& \mathrm{Im}\left\{a\bar{b}+c\bar{d}\right\} & \frac{1}{2}\left(|a|^2-|b|^2+|c|^2-|d|^2\right) \\ 
-\mathrm{Re}\left\{\bar{a}c+\bar{b}d\right\} & \mathrm{Re}\left\{\bar{a}d+\bar{b}c\right\} 
& -\mathrm{Im}\left\{a\bar{d}-b\bar{c}\right\} & - \mathrm{Re}\left\{\bar{a}c+\bar{b}d\right\}\\
\mathrm{Im}\left\{\bar{a}c+\bar{b}d\right\} & -\mathrm{Im}\left\{\bar{a}d+\bar{b}c\right\} 
& \mathrm{Re}\left\{a\bar{d}-b\bar{c}\right\} &  \mathrm{Im}\left\{\bar{a}c+\bar{b}d\right\}\\
\frac{1}{2}\left(|a|^2+|b|^2-|c|^2-|d|^2\right) & -\mathrm{Re}\left\{a\bar{b}-c\bar{d}\right\} 
& \mathrm{Im}\left\{a\bar{b}-c\bar{d}\right\} & \frac{1}{2}\left(|a|^2-|b|^2-|c|^2+|d|^2\right)
\end{matrix}\right).
\end{equation*}\\
At this stage, the relation between the Lorentz group and the sphere appears as a mere coincidence. 
In particular, since the original Lorentz group is defined by its linear action on a four-dimensional space, 
there is no reason for it to have anything to do with certain non-linear transformations of a two-dimensional 
manifold such as the sphere. However, it can be shown that this is not accidental. Following \cite{Oblak16}, 
we can suppose to perform a Lorentz transformation in Minkowski space-time 
equipped with standard coordinates $(t,x,y,z)$, i.e.
\begin{equation*}
x'^{\mu}=\Lambda^{\mu}_{\nu}x^{\nu},\hspace{1cm}\Lambda\in \mathscr{L}.
\end{equation*} 
We may introduce Bondi coordinates $(u,r,x^A)$ for Minkowski space-time as done in \eqref{eqn:135}. 
Then if we evaluate the limit for large values of the radial coordinate $r=\sqrt{x^2+y^2+z^2}$ 
keeping the value of $u=t-r$ fixed (i.e. on $\mathscr{I}$) and use the isomorphism \eqref{eqn:210} 
and hence \eqref{eqn:iso} we obtain the following behaviour:
\begin{equation*}
\zeta'=\frac{a\zeta+b}{c\zeta+d}+O(r^{-1}).
\end{equation*}
where
$\zeta=e^{i\phi}\cot\frac{\theta}{2}$. Furthermore it can be checked that both $u$ and $r$, under 
the effect of a Lorentz transformation on $\mathscr{I}$, undergo an angle-dependent rescaling. Hence 
we have obtained a fundamental result: Lorentz transformations acting on $\mathscr{I}$, expressed in 
terms of the parameters $a,b,c,d$ coincide with conformal transformations of $\mathrm{PSL}(2,\mathbb{C})$. 
Since asymptotically flat space-times have the same structure of a Minkowski space-time at infinity, 
this argument can be extended to all of them too. In the remainder we will use $\mathscr{L}$ to describe 
the group structure of $\mathscr{B}$, the isomorphism with $\mathrm{PSL}(2,\mathbb{C})$ being implicit.\\
Now we turn to analyse the transformations which involve a non-vanishing $\alpha(\theta,\phi)$. 
\begin{defn}$\\ $
\label{defn:suptra}
The Abelian subgroup of BMS transformations for which
\begin{equation}
\label{eqn:225}
\theta'=\theta,\hspace{0.5cm}\phi'=\phi,\hspace{0.5cm}u'=u+\alpha(\theta,\phi),
\end{equation}
is called \textit{supertranslation subgroup} and is denoted by $\mathscr{S}$.
\end{defn}
Under such a transformation the system of null hypersurfaces $u=\mathrm{const}$ is transformed into another 
system of null hypersurfaces $u'=\mathrm{const}$.  \\
To proceed further in the analysis of the structure of the BMS group we need to recall the concepts of 
\textit{right} and \textit{left} \textit{cosets} and, hence, that of \textit{normal subgroup} \cite{Algebra}. 
Consider a group $G$ and a subgroup $H$ of $G$. Introduce in $G$ the equivalence relation 
$\sim$ defined, for $g$, $a\in G$, as
\begin{equation*}
g\sim a \Longleftrightarrow ag^{-1}\in H\Longleftrightarrow a\in Hg.
\end{equation*}
It is easy to verify that the previous relation is reflexive, symmetric and transitive. 
\begin{defn}$\\ $
The equivalence class with respect to $\sim$ is called \textit{right coset of $H$ in $G$ with respect to $g$} 
and is denoted by $Hg$:
\begin{equation*}
[g]=\{hg : h\in H\}= Hg.
\end{equation*}
\end{defn}
Similarly, the \textit{left coset of $H$ in $G$ with respect to $g$} can be introduced as
\begin{equation*}
[g]^*=\{gh : h\in H\}=gH.
\end{equation*}
In general, right and left cosets are different sets.
\begin{defn}$\\ $
A subgroup  $N$ of $G$ which defines a unique partition,
\begin{equation*}
[g]=[g]^*\Longleftrightarrow gN=Ng\hspace{0.5cm}\forall g \in G
\end{equation*}
is called \textit{normal subgroup of} $G$. 
\end{defn}
Clearly it follows that for every $n\in N$ and $g\in G$ the product $gNg^{-1}\subseteq N$. Note that every group 
$G$ possesses normal subgroups, since $G$ and the identity are normal subgroups. \\
Now consider for a general subgroup $H$ of $G$ the \textit{quotient group} (or \textit{factor group}) $G/H$, defined as
\begin{equation*}
G/H=\{[g] : g\in G\}.
\end{equation*}
If $H$ is normal the elements of $G/H$ are, indistinctly, the right and left cosets. Furthermore, 
under this hypothesis, the set $G/H$ can be equipped with a group structure in a natural way by 
defining the product $\ast$:
\begin{align*}
\ast : &G/H\times G/H\longrightarrow G/H \\
&gH\ast g'H \equiv gg'H,
\end{align*}
i.e., 
\begin{equation*}
[g]\ast [g']\equiv [gg'].
\end{equation*}
It can be shown that $G/H$ equipped with the product $\ast$ satisfies the group axioms. We are now ready to 
discuss further the BMS properties.\\
Any element $b$ of $\mathscr{B}$ can be written as
\begin{equation*}
b=(\Lambda,\alpha).
\end{equation*}
Note that with this nomenclature any element $\Lambda$ of $\mathscr{L}$ (or, equivalently, of 
$\mathrm{PSL}(2,\mathbb{C})$) can be written as $\Lambda=(\Lambda,0)$ and any element $s$ of $\mathscr{S}$ 
as $s=(\mathbb{I},\alpha)$, where $\mathbb{I}$ denotes the identity in $\mathscr{L}$.\\ 
The action of $b$ on the variables $(\zeta,u)$ is 
\begin{equation*}
b(\zeta,u)=(f(\zeta),K[u+\alpha(\zeta,\bar{\zeta})]),
\end{equation*}
where $f$ is the element of $\mathrm{PSL}(2,\mathbb{C})$ which corresponds to $\Lambda$ through the above 
discussed isomorphism and $K$ is its conformal factor.\\
It is easy to show that, with this notation, we have for the inverse of $b$:
\begin{equation}
\label{eqn:208}
b^{-1}=(\Lambda^{-1},-K\alpha).
\end{equation}
On considering an element $s=(\mathbb{I},\beta)$ of $\mathscr{S}$ we have
\begin{equation*}
bsb^{-1}(\zeta,u)=(\zeta,u+K\beta(\zeta,\bar{\zeta}))=s'(\zeta,u),
\end{equation*}
with 
\begin{equation*}
s'=(\mathbb{I},K\beta)\in\mathscr{S}.
\end{equation*}
From the above discussion we have the following
\begin{prop}$\\ $
\label{prop:BMS1}
The supertranslations $\mathscr{S}$ form an Abelian normal, infinite-parameter, subgroup of the BMS group:
\begin{equation*}
b\mathscr{S}b^{-1}=\mathscr{S}\hspace{0.5cm}\mathrm{for}\hspace{1.2mm} \mathrm{all}\hspace{1.2mm} b\in\mathscr{B}.
\end{equation*}
\end{prop}
Under the assumption that the function $\alpha$ is twice differentiable, we can expand it into spherical harmonics as
\begin{equation}
\label{eqn:206}
\alpha(\theta,\phi)=\sum_{l=0}^{\infty}\sum_{m=-l}^{l}\alpha_{l,m}Y_{l,m}(\theta,\phi),\hspace{1cm}\alpha_{l,-m}
=(-1)^m\bar{\alpha}_{l,m}.
\end{equation}
\begin{defn}$\\ $
If in decomposition \eqref{eqn:206} $\alpha_{l,m}=0$ for $l>2$, i.e.
\begin{equation}
\label{eqn:207}
\alpha\equiv\alpha_t=\epsilon_0+\epsilon_1\sin\theta\cos\phi+\epsilon_2\sin\theta\sin\phi+\epsilon_3\cos\theta,
\end{equation}
then the supertranslations reduce to a special case, called \textit{translation} subgroup, denoted by 
$\mathscr{T}$, with just four independent parameters $\epsilon_0,...,\epsilon_3$. 
\end{defn}
It is easy to show that $\zeta=e^{i\phi}\cot\frac{\theta}{2}$ implies
\begin{equation*}
\cos\phi=\frac{\zeta+\bar{\zeta}}{2\sqrt{\zeta\bar{\zeta}}},\hspace{1cm}\sin\phi
=\frac{i(\bar{\zeta}-\zeta)}{2\sqrt{\zeta\bar{\zeta}}},
\end{equation*}
\begin{equation*}
\cos\theta=\frac{\zeta\bar{\zeta}-1}{1+\zeta\bar{\zeta}},\hspace{1cm}\sin\theta
=\frac{2\sqrt{\zeta\bar{\zeta}}}{1+\zeta\bar{\zeta}}.
\end{equation*}
Then equation \eqref{eqn:207} becomes
\begin{equation*}
\alpha_t=\epsilon_0+\epsilon_1\frac{\zeta+\bar{\zeta}}{1+\zeta\bar{\zeta}}+\epsilon_2
\frac{(i\zeta-i\bar{\zeta})}{1+\zeta\bar{\zeta}}+\epsilon_3\frac{\zeta\bar{\zeta}-1}{1+\zeta\bar{\zeta}}
\end{equation*}
\begin{equation*}
=\frac{A+B\zeta+\bar{B}\bar{\zeta}+C\zeta\bar{\zeta}}{1+\zeta\bar{\zeta}},
\end{equation*}
with $A$ and $C$ real. 
Hence in terms of $\zeta$ and $\bar{\zeta}$ a translation is
\begin{equation*}
u=t-r\rightarrow u'=u+\frac{A+B\zeta+\bar{B}\bar{\zeta}+C\zeta\bar{\zeta}}{1+\zeta\bar{\zeta}},
\end{equation*}
\begin{equation*}
\zeta\rightarrow\zeta'=\zeta.
\end{equation*}
If we let $t,x,y,z$ be Cartesian coordinates in Minkowski space-time, it is easy to see that
\begin{equation*}
Z^2\zeta=\frac{(x+iy)(1-z/r)}{4r},\hspace{1cm}x=r(\zeta+\bar{\zeta})Z,
\end{equation*}
\begin{equation*}
y=-ir(\zeta-\bar{\zeta})Z,\hspace{1cm}z=r(\zeta\bar{\zeta}-1)Z,
\end{equation*}
where $Z=1/(1+\zeta\bar{\zeta})$. Now if we perform a translation
\begin{equation*}
t\rightarrow t'=t+a,\hspace{0.5cm}x\rightarrow x'=x+b,\hspace{0.5cm}y\rightarrow y'=y+c,
\hspace{0.5cm}z\rightarrow z'=z+d,
\end{equation*}
it is easy to get 
\begin{align*}
&u=t-r\rightarrow u'=u+Z(A+B\zeta+\bar{B}\bar{\zeta}+C\zeta\bar{\zeta})+O(r^{-1}),\\
&\zeta\rightarrow\zeta'=\zeta+O(r^{-1}).
\end{align*}
with $A=a+d$, $B=b-ic$ and $C=a-d$. Thus, the nomenclature \lq translation\rq\hspace{0.1mm} is consistent 
with that for the space-time translations in Minkowski space-time. In fact we have just shown that any 
translation in the ordinary sense induces a translation (i.e. an element of $\mathscr{T}$) on $\mathscr{I}^+$.\\
It is easy to verify that for any $b=(\Lambda,\alpha)\in\mathscr{B}$ and for any 
$t=(\mathbb{I},\alpha_t)\in\mathscr{T}$ we have 
\begin{equation*}
btb^{-1}(\zeta,u)=(\zeta,u+K\alpha_t)=t'(\zeta,u),
\end{equation*}
with 
\begin{equation*}
t'=(\mathbb{I},K\alpha_t(\zeta,\bar{\zeta}))\in\mathscr{T}.
\end{equation*}
Note that it is not obvious that $K\alpha_t$ is still a function of $\theta$ and $\phi$ containing only 
zeroth- and first-order harmonics. A proof of this result will be given in Sect. \ref{sect:8}. 
On taking for the moment this result for true, the following proposition holds:
\begin{prop}$\\ $
The translations $\mathscr{T}$ form a normal four-dimensional subgroup of $\mathscr{B}$:
\begin{equation*}
b\mathscr{T}b^{-1}=\mathscr{T}\hspace{0.5cm}\mathrm{for}\hspace{1.2mm} \mathrm{all}\hspace{1.2mm} b\in\mathscr{B},
\end{equation*}
and clearly
\begin{equation*}
s\mathscr{T}s^{-1}=\mathscr{T}\hspace{0.5cm}\mathrm{for}\hspace{1.2mm} \mathrm{all}\hspace{1.2mm} s\in\mathscr{S}.
\end{equation*}
\end{prop}
We have the following inclusion relations:
\begin{equation*}
\mathscr{T}\subset\mathscr{S}\subset\mathscr{B}.
\end{equation*}
The next step will be to investigate the group structure of $\mathscr{B}$. It is easy to show 
that for any $b\in\mathscr{B}$ there exists a unique $\Lambda\in\mathscr{L}$ and $s\in\mathscr{S}$ 
such that $b=\Lambda s$. In fact given 
\begin{equation*}
\Lambda=(\Lambda,0)\in\mathscr{L},\hspace{1cm}s=(\mathbb{I},\alpha)\in\mathscr{S},
\end{equation*}
we have that
\begin{equation*}
\Lambda s(\zeta,u)=\Lambda(\zeta,u+\alpha)=(f(\zeta),K[u+\alpha(\zeta,\bar{\zeta})])=b(\zeta,u)
\end{equation*}
with $b=(\Lambda,\alpha)$. The uniqueness results from the observation that $\mathscr{L}\cap\mathscr{S}=\{e\}$ where 
$e=(\mathbb{I},0)$ is the identity in $\mathscr{B}$, since if $fs=f's'$, then $f'^{-1}f=s's^{-1}\in\mathscr{L}\cap\mathscr{S}$ 
implying $f'=f$ and $s'=s$. Hence we have
\begin{equation}
\label{eqn:233}
\mathscr{B}=\mathscr{L}\mathscr{S}.
\end{equation}
Furthermore, the supertranslations $\mathscr{S}$ form an (Abelian) normal subgroup of $\mathscr{B}$, according 
to \ref{prop:BMS1}. Thus we can already state that, by definition of semi-direct product, 
\begin{prop}$\\ $
The BMS group is a semi-direct product of the conformal group of the unit $2$-sphere with the 
supertranslations group, i.e.
\begin{equation*}
\mathscr{B}=\mathscr{L}\rtimes\mathscr{S}.
\end{equation*}
\end{prop}
We can say more by specifying an action of $\mathscr{L}$ on $\mathscr{S}$ and, hence, by specifying a product rule 
for two elements of $\mathscr{B}$. Let $\mathscr{S}$ be the vector space of real functions on the Riemann sphere. 
Let $\sigma$ be a smooth right action of $\mathscr{L}$ on $\mathscr{S}$ defined as 
\begin{subequations}
\label{eqn:212}
\begin{equation}
\sigma:(\Lambda,\alpha)\in\mathscr{L}\times\mathscr{S}\longrightarrow\sigma_{\Lambda}(\alpha)
\equiv \alpha\Lambda\in\mathscr{S}
\end{equation}
such that 
\begin{equation}
\alpha(\zeta,\bar{\zeta})\Lambda=K^{-1}\alpha(f(\zeta),\bar{f}(\bar{\zeta})),
\end{equation}
\end{subequations}
where $K$ is the conformal factor associated with $f$, the element of $\mathrm{PSL}(2,\mathbb{C})$ that 
corresponds to $\Lambda$. Then it is easy to verify that the composition law for the elements of $\mathscr{B}$ is 
\begin{equation*}
b_1\cdot b_2=(\Lambda_1,\alpha_1)\cdot(\Lambda_2,\alpha_2)=(\Lambda_1\cdot\Lambda_2,\alpha_2+\alpha_1f_2).
\end{equation*}
Note that the inverse of $b$ in \eqref{eqn:208} may be written as
\begin{equation*}
b^{-1}=(\Lambda^{-1},\left[\alpha \Lambda^{-1}\right]^{-1}).
\end{equation*}
Thus the BMS group is the \textit{right} semi-direct product \cite{Col14} of $\mathscr{L}$ with $\mathscr{S}$ 
under the action $\sigma$, i.e.
\begin{equation}
\label{eqn:209}
\mathscr{B}=\mathscr{L}\rtimes_{\sigma}\mathscr{S}.
\end{equation}
Historically, this semi-direct product structure was realized by \cite{Cant}. Succesively this idea was developed 
by \cite{GerBMS} who gave an incorrect formula for the action \eqref{eqn:212}. Eventually the mistake was amended 
by \cite{MC2}, who defined a good action to describe the semi-direct product structure of the BMS group. 
However, the idea used here of giving a right action and hence of describing the BMS group as a \textit{right} 
semi-direct product was not developed by any of these authors and it is an original contribution of our work.\\
Furthermore, it can be shown that from the above discussion it follows that, by virtue 
of the first isomorphism theorem,
\begin{equation}
\label{eqn:234}
\mathscr{L}\cong\mathscr{B}/\mathscr{S},
\end{equation}
i.e. $\mathscr{L}$ is the factor group of $\mathscr{B}$ with respect to its normal subgroup $\mathscr{S}$. 
\begin{oss}$\\ $
Note that the structure of the BMS group is similar to that of the Poincar\'e group, denoted by $\mathscr{P}$. 
In fact the Poincar\'e group can be expressed as the semi-direct product of the connected component 
of the Lorentz group $\mathscr{L}$ and the translations group $\mathscr{T}$, the former being the factor group of 
$\mathscr{P}$ with respect to the latter, i.e. $\mathscr{L}\cong\mathscr{P}/\mathscr{T}$. The action of 
$\mathscr{L}$ on $\mathscr{T}$ is the \lq natural\rq\hspace{0.1mm} one, i.e. the usual multiplication of 
an element $\Lambda\in \mathscr{L}$ with a vector $b\in\mathscr{T}$.
\end{oss}
\begin{thm}$\\ $
If $N'$ is a 4-dimensional normal subgroup of $\mathscr{B}$, then $N'$ is contained in $\mathscr{S}$.
\end{thm}
\begin{proof}
Consider the image $N'/\mathscr{S}$ of $N'$ under the homomorphism $\mathscr{B}\rightarrow\mathscr{B}/\mathscr{S}$. 
Since $N'$ by hypothesis is a normal subgroup of $\mathscr{B}$, $N'/\mathscr{S}$ is a normal subgroup of 
$\mathscr{B}/\mathscr{S}$ and hence, by proposition \ref{prop:BMS1}, a subgroup of the connected component 
of the Lorentz group $\mathscr{L}$. However, the only normal subgroups of $\mathscr{L}$ are $\mathscr{L}$ 
itself and the identity $e$ of $\mathscr{L}$. Then $N'$ must be $6$-dimensional, contrary to hypothesis. 
Therefore $N'/\mathscr{S}=e$; $N'$ is thus contained in $\mathscr{S}$. 
\end{proof}
\section{BMS Lie Algebra}
\label{sect:7}
In this section we are going to investigate the Lie Algebra of the BMS group. At first we consider the 
generators of $\mathrm{PSL}(2,\mathbb{C})$. For an infinitesimal conformal transformation we know 
that the $x^A$ coordinates change as
\begin{equation*}
x^A\rightarrow x'^A=x^A+f^A,
\end{equation*}
where $f^A$ is a conformal Killing vector of the unit 2-sphere. Furthermore, from 
\eqref{eqn:220} and taking into account \eqref{eqn:214} an infinitesimal transformation for $u$ reads as
\begin{equation*}
u\rightarrow u'=Ku=e^{\frac{1}{2}D_Af^A}u\simeq u+\frac{u}{2}D_Af^A.
\end{equation*}
Thus, the generator of transformation \eqref{eqn:220} is
\begin{equation}
\label{eqn:221}
\xi_R=f^A\partial_A+\frac{u}{2}D_Af^A\partial_u.
\end{equation}
To see how their Lie algebra closes, consider the Lie bracket of two of them, $\xi_{R_1}$ and $\xi_{R_2}$:
\begin{equation*}
[\xi_{R_1},\xi_{R_2}]=\left[f_1^A\partial_A+\frac{u}{2}D_Bf_1^B\partial_u,f_2^C\partial_C
+\frac{u}{2}D_Cf_2^C\partial_u\right]
\end{equation*}
\begin{equation*}
=(f_1^A\partial_Af_2^C-f_2^A\partial_Af_1^C)\partial_C+\frac{u}{2}(f_1^A\partial_AD_Cf_2^C
-f_2^C\partial_CD_Bf_1^B)\partial_u.
\end{equation*}
After some calculation the term proportional to $\partial_u$ becomes
\begin{equation*}
\frac{u}{2}D_C(f_1^A\partial_Af_2^C-f_2^A\partial_Af_1^C)+\frac{u}{2}(\partial_A\Gamma^{C}{}_{CB})
(f_1^Af_2^B-f_1^Bf_2^A).
\end{equation*}
The last term in the previous equation vanishes since it can be shown by direct calculation 
that for the metric $q_{AB}$
\begin{equation*}
\partial_A\Gamma^{C}{}_{CB}=-\frac{1}{\sin^2\theta}\delta^{\theta}_{A}\delta_B^{\theta},
\end{equation*}
and hence
\begin{equation*}
\left(\partial_A\Gamma^{C}{}_{CB}\right)(f_1^Af_2^B-f_1^Bf_2^A)=-\frac{1}{\sin^2\theta}(f_1^{\theta}f_2^{\theta}
-f_1^{\theta}f_2^{\theta})=0.
\end{equation*}
Finally we get 
\begin{subequations}
\label{eqn:222}
\begin{equation}
\label{eqn:223}
[\xi_{R_1},\xi_{R_2}]=\xi_{\hat{R}}=\hat{f}^A\partial_A+\frac{u}{2}D_A\hat{f}^A\partial_u
\end{equation}
where
\begin{equation}
\label{eqn:224}
\hat{f}^A=f_1^B\partial_Bf_2^A-f_2^B\partial_Bf_1^A.
\end{equation}
\end{subequations}
We take now into account the generators of supertranslations. It is clear from \eqref{eqn:225} that these are
\begin{equation*}
\xi_T=\alpha\partial_u.
\end{equation*}
where $\alpha$ is an arbitrary function of $\theta$ and $\phi$. It follows that the Lie bracket of two 
generators, $\xi_{T_1}$ and $\xi_{T_2}$ vanish, i.e.
\begin{equation}
\label{eqn:228}
[\xi_{T_1},\xi_{T_2}]=0,
\end{equation}
that is just a restatement that the supertranslation group is Abelian. The only thing left to do is to 
calculate the Lie bracket of $\xi_{R}$ and $\xi_{T}$. It is easy to see that 
\begin{subequations}
\label{eqn:226}
\begin{equation}
\label{eqn:227}
[\xi_R,\xi_T]=[f^A\partial_A+\frac{u}{2}D_Bf^B\partial_u,\alpha\partial_u]
=\xi_{\hat{T}}=\hat{\alpha}\partial_u,
\end{equation}
where
\begin{equation}
\label{eqn:228}
\hat{\alpha}=f^A\partial_A\alpha-\frac{\alpha}{2}D_Bf^B.
\end{equation}
\end{subequations}
If we consider now $\xi$ as defined in \eqref{eqn:158} it turns out that $\xi=\xi_R+\xi_T$. 
From the above discussions one obtains that
\begin{eqnarray}
\label{eqn:229}
[\xi_1,\xi_2]&=& [\xi_{R_1},\xi_{R_2}]+[\xi_{R_1},\xi_{T_2}]+[\xi_{T_1},\xi_{R_2}]
\nonumber \\
&=& \hat{f}^A\partial_A+\frac{u}{2}D_A\hat{f}^A\partial_u+(\hat{\alpha}_{2}-\hat{\alpha}_1)\partial_u,
\end{eqnarray}
with 
\begin{equation*}
\hat{\alpha}_{2}=f^A_1\partial_A\alpha_2-\frac{\alpha_2}{2}D_Bf^B_1,\hspace{1cm}\hat{\alpha}_{1}
=f^A_2\partial_A\alpha_1-\frac{\alpha_1}{2}D_Bf^B_2.
\end{equation*}
To sum up, the Lie algebra of the BMS group, $\mathfrak{bms_4}$, is
\begin{align*}
&[\xi_{R_1},\xi_{R_2}]=\xi_{\hat{R}},\hspace{0.5cm}\mathrm{with}\hspace{0.5cm}\hat{f}^A
=f_1^B\partial_Bf_2^A-f_2^B\partial_Bf_1^A; \\
&[\xi_{T_1},\xi_{T_2}]=0;\\
&[\xi_R,\xi_T]=\xi_{\hat{T}},\hspace{0.5cm}\mathrm{with}\hspace{0.5cm}\hat{\alpha}
=f^A\partial_A\alpha-\frac{\alpha}{2}D_Bf^B.
\end{align*}
Since, as shown in the previous section, the BMS group is a semi-direct product, it follows that the 
BMS Lie algebra, $\mathfrak{bms_4}$, should be taken to be the semi-direct sum of the Lie algebra of 
conformal Killing vectors $X=f^A\partial_A$ of the Riemann sphere, denoted by $\mathfrak{so(3,1)}$ (since it 
can be taken to be the algebra of $\mathscr{L}$) with that of the functions $\alpha(x^A)$ on the Riemann 
sphere, which we denote by $\mathscr{S}$, the supertranslation group being Abelian. Given an element 
$X=f^A\partial_A\in\mathfrak{so(3,1)}$ ($f^A$ being a generator of conformal transformations in \eqref{eqn:157}) 
we know that the exponential map associated to it, $e^{X}$, is an element of $\mathscr{L}$. 
Then consider the 1-parameter group of transformations in $\mathscr{S}$ defined as
\begin{equation}
\label{eqn:219}
\sigma_{e^{tX}}(\alpha)=\alpha e^{tX},
\end{equation}
where $\sigma$ is that of \eqref{eqn:212}. Consider the map
\begin{equation*}
\Sigma:f^A\partial_A\in\mathfrak{so(3,1)}\longrightarrow\Sigma_{f^A\partial_A}\in\mathrm{End}\mathscr{S},
\end{equation*}
such that 
\begin{equation*}
\Sigma_{f^A\partial_A}:\alpha\in\mathscr{S}\longrightarrow\Sigma_{f^A\partial_A}(\alpha)=\frac{d}{dt}	
\left.(\sigma_{e^{tf^A\partial_A}}(\alpha))\right|_{t=0}\in\mathscr{S}.
\end{equation*}
\\
Note that $\Sigma_{f^A\partial_A}(\alpha)$ is the infinitesimal generator of \eqref{eqn:219}. 
Hence \cite{OblakBMS} we have 
\begin{equation}
\label{eqn:213}
\mathfrak{bms_4}=\mathfrak{so(3,1)}\oplus_{\Sigma}\mathscr{S}.
\end{equation}
The Lie algebra $\mathfrak{bms_4}$ is determined by three arbitrary functions $f^A$ and $\alpha$ on the circle. 
Thus, defining $X=f^A\partial_A$ and labelling elements of \eqref{eqn:213} as pairs $(X,\alpha)$, 
we know that the Lie bracket in $\mathfrak{so(3,1)}\oplus_{\Sigma}\mathscr{S}$ are
\begin{equation}
\label{eqn:215}
[(X_1,\alpha_1),(X_2,\alpha_2)]=([X_1,X_2],\Sigma_{f_1^A\partial_A}(\alpha_2)-\Sigma_{f^B_2\partial_B}(\alpha_1)).
\end{equation}
Equation \eqref{eqn:215} follows from the fact that $\mathscr{S}$ is Abelian, otherwise there would be an extra 
term involving the commutator of the two elements $\alpha_1$ and $\alpha_2$.
Since we have, using \eqref{eqn:212} and \eqref{eqn:214}
\begin{equation*}
\Sigma_{f^A\partial_A}(\alpha)(x^B)=\frac{d}{dt}\left.(K^{-1}_{e^{tf^A\partial_A}}\alpha
(e^{tf^C\partial_C}x^B))\right|_{t=0}=\frac{d}{dt}\left.(e^{-\frac{1}{2}tD_Af^A}
\alpha(e^{tf^C\partial_C}x^B))\right|_{t=0}
\end{equation*}
\begin{equation*}
=-\frac{\alpha}{2}D_Af^A+ f^B\partial_B\alpha,
\end{equation*}
then \eqref{eqn:215} may be written as
\begin{subequations}
\label{eqn:230}
\begin{equation}
\label{eqn:231}
[(X_1,\alpha_1),(X_2,\alpha_2)]=(\hat{X},\hat{\alpha}),
\end{equation}
with 
\begin{equation}
\label{eqn:232}
\hat{f}^A=f_1^B\partial_Bf_2^A-f_2^B\partial_Bf_1^A,\hspace{1cm}\hat{\alpha}=f_1^B\partial_B\alpha_2
-\frac{\alpha_2}{2}D_Af^A_1-(1\leftrightarrow2),
\end{equation}
\end{subequations}
as remarked in \cite{Barn2010,Barn2010a}.
\begin{oss}$\\ $
Note that this result, obtained from the theory of semi-direct product of groups and their Lie algebra, 
is in complete agreement with that obtained just by looking at the generators, expressed in \eqref{eqn:229}. 
Note also that $\hat{f}^A$ of \eqref{eqn:232} coincides with that of \eqref{eqn:224} and that 
$\hat{\alpha}=\hat{\alpha}_2-\hat{\alpha}_1$.
\end{oss}
Depending on the space of functions under consideration, there are many options which define what is actually 
meant by $\mathfrak{bms_4}$. The approach that will be followed in this work is originally due to 
\cite{Sachs1} and successively amended by \cite{Ant91}. Another approach, based on the Virasoro algebra, 
can be found in \cite{Barn2010}. \\
In general, we consider any $S$-dimensional Lie transformation group of a $R$-dimensional space. Let 
the coordinates of the space be $y^{\alpha}$ $(\alpha=1,..,R)$ and the parameters of the group be 
$z^{\mu}$ $(\mu=1,...,S)$, where $z^{\mu}=0$ is the identity of the group. Then the transformations have the form 
\begin{equation*}
y'^{\alpha}=f^{\alpha}(y^{\beta};z^{\mu}),\hspace{1cm}\mathrm{where}\hspace{1cm}f^{\alpha}(y^{\beta};0)=y^{\alpha}.
\end{equation*}
The functions $f^{\alpha}$ are assumed to be twice differentiable. The $S$ generators of the group are 
the vector fields given by
\begin{equation}
\label{eqn:217}
P_{\mu}=\left.\frac{\partial f^{\alpha}}{\partial z^{\mu}}\right|_{z^{\mu}=0}\frac{\partial}{\partial y^{\alpha}}.
\end{equation}
Applying these ideas to the BMS group,with the Sachs notation \cite{Sachs1}, one finds for the supertranslations, using the expansion \eqref{eqn:206}:
\begin{equation*}
P_{l,m}=Y_{l,m}(\theta,\phi)\frac{\partial}{\partial u},\hspace{1cm}P_{l,m}=(-1)^m\bar{P}_{l,-m},
\end{equation*}
and hence
\begin{equation*}
[P_{l,m},P_{n,r}]=0,
\end{equation*}
i.e. two supertranslations commute. \\
To find the generators of conformal transformations we have to be careful. We know that any conformal 
transformation has the form 
\begin{equation*}
\zeta'=\frac{a\zeta+b}{c\zeta+d},\hspace{1cm}\zeta=e^{i\phi}\cot\frac{\theta}{2}.
\end{equation*}
By direct calculations one obtains that the following equations hold for $\theta'$, $\phi'$ and $u'$:
\begin{subequations}
\label{eqn:216}
\begin{align}
&\theta'=2\arctan\left[\left(\frac{|c\zeta+d|^2}{|a\zeta+b|^2}\right)^{1/2}\right],\\ \nonumber\\
&\phi'=\arctan\left[\frac{\mathrm{Im}\{(a\zeta+b)(\bar{c}\bar{\zeta}+\bar{d})\}}
{\mathrm{Re}\{(a\zeta+b)(\bar{c}\bar{\zeta}+\bar{d})\}}\right],\\ \nonumber\\
&u'=\frac{1+|\zeta|^2}{|a\zeta+b|^2+|c\zeta+d|^2}u.
\end{align}
\end{subequations}
On denoting by $x$ the parameter of the transformation it is clear that $a,b,c,d$ are functions of $x$ such 
that $a(0)=d(0)=1$ and $c(0)=b(0)=0$. We have to apply \eqref{eqn:217} to \eqref{eqn:216}. 
It is easy to verify that 
\begin{align*}
\left.\frac{d\theta'}{dx}\right|_{x=0}&=\frac{\cos^3\frac{\theta}{2}}{\sin\frac{\theta}{2}}
\frac{d}{dx}\left.\left[\frac{|c|^2|\zeta|^2+|d|^2+c\bar{d}\zeta+\bar{c}d\bar{\zeta}}
{|a|^2|\zeta|^2+|b|^2+a\bar{b}\zeta+\bar{a}b\bar{\zeta}}\right]\right|_{x=0},\\ \\
\left.\frac{d\phi'}{dx}\right|_{x=0}&=\cos^2\phi\frac{d}{dx}\left.\left[\frac
{\mathrm{Im}\{a\bar{c}|\zeta|^2+a\bar{d}\zeta+b\bar{c}\bar{\zeta}+b\bar{d}\}}
{\mathrm{Re}\{a\bar{c}|\zeta|^2+a\bar{d}\zeta+b\bar{c}\bar{\zeta}+b\bar{d}\}}\right]\right|_{x=0},\\ \\
\left.\frac{du'}{dx}\right|_{x=0}&=\sin^2\frac{\theta}{2}\frac{d}{dx}\left.\left[(|a|^2+|c|^2)|\zeta|^2
+(a\bar{b}+c\bar{d})\zeta+(b\bar{a}+d\bar{c})\bar{\zeta}+|b|^2+|d|^2\right]\right|_{x=0}u.
\end{align*}
These equations hold in general for any conformal transformation. We choose now to work with Lorentz 
transformations, and thus to use Lorentz generators $L_i$ and $R_i$ of rotations and boosts, respectively. 
To know the coefficients $a,b,c,d$ corresponding to a Lorentz transformation we need to use the isomorphism 
\eqref{eqn:210}. Any rotation of an angle $\varphi$ about an axis $\hat{n}$ and any boost of rapidity 
$\chi$ about an axis $\hat{e}$ can be performed by using a $\mathrm{SL}(2,\mathbb{C})$ matrix given by
\begin{subequations}
\label{eqn:236}
\begin{align}
&U_{\hat{n}}(\varphi)=e^{\frac{i}{2}\varphi\hat{n}\cdot\vec{\sigma}}=\mathbb{I}\cos\frac{\varphi}{2}
+i\hat{n}\cdot\vec{\sigma}\sin\frac{\varphi}{2},\\
&H_{\hat{e}}(\chi)=e^{\frac{1}{2}\chi\hat{e}\cdot\vec{\sigma}}=\mathbb{I}\cosh\frac{\chi}{2}
+\hat{e}\cdot\vec{\sigma}\sinh\frac{\chi}{2},
\end{align}
\end{subequations}
respectively, where $\vec{\sigma}=(\sigma_x,\sigma_y,\sigma_z)$ are the Pauli matrices. The parameter 
$x$ of the two transformations is $\varphi$ and $\chi$, respectively. After some calculations we find that 
the vector fields that generate the transformations are 
\begin{align}
L^{23}=L_x&=-\sin\phi\frac{\partial}{\partial\theta}-\cot\theta\cos\phi\frac{\partial}{\partial \phi},\\
L^{13}=L_y&=-\cos\phi\frac{\partial}{\partial\theta}+\cot\theta\sin\phi\frac{\partial}{\partial \phi},\\
L^{12}=L_z&=\frac{\partial}{\partial \phi},\\
L^{10}=R_x&=\cos\theta\cos\phi\frac{\partial}{\partial\theta}-\frac{\sin\phi}{\sin\theta}\frac{\partial}{\partial\phi}
-u\sin\theta\cos\phi\frac{\partial}{\partial u},\\
L^{20}=R_y&=-\cos\theta\sin\phi\frac{\partial}{\partial\theta}-\frac{\cos\phi}{\sin\theta}\frac{\partial}{\partial\phi}
+u\sin\theta\sin\phi\frac{\partial}{\partial u},\\
L^{30}=R_z&=-\sin\theta\frac{\partial}{\partial \theta}-u\cos\theta\frac{\partial}{\partial u}.
\end{align}
Note that rotations are characterized by $K=1$. The $\{P_{l,m}\}$
and $\{L^{ab}\}$ form a complete set of linearly independent vector fields for the Lie algebra 
$\mathfrak{bms}_4$. We can find now the commutators
\begin{equation*}
[L^{ab},L^{cd}]=\eta^{ac}L^{bd}+\eta^{bd}L^{ac}-\eta^{ad}L^{bc}-\eta^{bc}L^{ad},
\end{equation*}
\begin{equation*}
[L_i,L_j]=\epsilon_{ijk}L_k,\hspace{1cm}[R_i,R_j]=-\epsilon_{ijk}L_k,\hspace{1cm}[L_i,R_j]=-\epsilon_{ijk}R_k,
\end{equation*}
where $\eta^{ab}=\mathrm{diag}(1,-1,-1,-1)$ and $\epsilon_{ijk}$ is the Levi-Civita symbol. Note that we have 
just obtained the classical Lorentz algebra. Furthermore, it is easy to derive the following commutator:
\begin{equation}
\label{eqn:218}
\left[L^{ab},\alpha(\theta,\phi)\frac{\partial}{\partial u}\right]=\left[L^{ab}\alpha(\theta,\phi)
-\alpha(\theta,\phi)W(L^{ab})\right]\frac{\partial}{\partial u},
\end{equation}
where $W(L^{ab})$ is defined by the relation 
\begin{equation*}
\frac{\partial}{\partial u}(L^{ab}f)=L^{ab}\frac{\partial f}{\partial u}+W(L^{ab})\frac{\partial f}{\partial u},
\end{equation*}
for arbitrary $f(u)$. \\
For convenience we introduce the raising and lowering operators,
\begin{align*}
L^{\pm}&=L_y\pm iL_x=-e^{\pm i\phi}\left(\frac{\partial}{\partial \theta}
\pm i\cot\theta\frac{\partial}{\partial \phi}\right),\\
R^{\pm}&=R_x\mp i R_y=e^{\pm i\phi}\left(\cos\theta\frac{\partial}{\partial \theta}\pm\frac{i}{\sin\theta}
\frac{\partial}{\partial \phi}-u\sin\theta\frac{\partial}{\partial u}\right), 
\end{align*}
in terms of which, using equation \eqref{eqn:218}, we give the commutation relations 
with the generators of supertranslations:
\begin{align*}
&[L_z,P_{l,m}]=imP_{l,m},\\ \nonumber\\
&[L^+,P_{l,m}]=-\sqrt{(l-m)(l+m+1)}P_{l,m+1},\\ \nonumber\\
&[L^-,P_{l,m}]=\sqrt{(l+m)(l-m+1)}P_{l,m-1},
\end{align*}
\begin{align*}
&[R_z,P_{l,m}]=-(l-1)\sqrt{\frac{(l+m+1)(l-m+2)}{(2l+1)(2l+3)}}P_{l+1,m}\\
&+(l+2)\sqrt{\frac{(l+m)(l-m)}{4l^2-1}}P_{l-1,m},\\ \\
&[R^+,P_{l,m}]=(l-1)\sqrt{\frac{(l+m+1)(l+m+2)}{(2l+1)(2l+3)}}P_{l+1,m+1}\\
&+(l+2)\sqrt{\frac{(l-m-1)(l-m)}{4l^2-1}}P_{l-1,m+1},\\ \\
&[R^-,P_{l,m}]=-(l-1)\sqrt{\frac{(l-m+1)(l-m+2)}{(2l+1)(2l+3)}}P_{l+1,m-1}\\
&-(l+2)\sqrt{\frac{(l+m-1)(l+m)}{4l^2-1}}P_{l-1,m-1}.
\end{align*}
The form of the commutation relations shows that the BMS algebra is the semi-direct sum of the Lorentz algebra 
$\mathfrak{so(3,1)}$ with the infinite Lie algebra $\mathscr{T}$, as remarked before.
\section{Good and bad cuts}
\label{sect:8}
We begin this section by citing a remarkable result obtained by Sachs.
\begin{thm}{\cite{Sachs1}}$\\ $
\label{thm:transl}
The only $4$-dimensional normal subgroup of the BMS group is the translation group.
\end{thm}
Theorem \ref{thm:transl} characterizes translations uniquely: the translation normal subgroup of the BMS group 
is singled out by its group-theoretic properties. Since we have shown that the translations $\mathscr{T}$ are 
the BMS transformations induced on $\mathscr{I}^+$ by translations in Minkowski space-time, theorem 
\ref{thm:transl} makes it possible for  us to define the asymptotic translations 
of a general asymptotically flat space-time 
as the BMS elements belonging to this normal subgroup. However a similar procedure for $\mathscr{L}$, 
i.e. rotations and boosts, fails. Thus, as we will discuss in this section, there are several problems in 
identifying the Poincar\'e group as a subgroup of $\mathscr{B}$.\\
The Poincar\'e group is the symmetry group of flat space-time, hence it might have been thought that a suitably 
asymptotically flat space-time should, in some appropriate sense, have the Poincar\'e group as an asymptotic 
symmetry group. Instead, it turns out that in general we seem only to obtain the BMS group (which has  
the unpleasant feature of being an infinite-dimensional group) as the asymptotic symmetry group of an 
asymptotically flat space-time. \\
To better understand the nature of this problem we revert to Minkowski space-time and see how the Poincar\'e 
group arises in that case as a subgroup of the BMS group. The BMS group was defined as the group of 
transformations which conformally preserves the induced metric on $\mathscr{I}^+$ and the strong conformal 
geometry. However, the BMS group is much larger than the Poincar\'e group and thus the former must preserve 
less structure on $\mathscr{I}^+$ than does the latter. The preservation of this additional structure, 
in the case of Minkowski space-time, should allow us to restrict the BMS transformations to Poincar\'e 
transformations, since we know that $\mathscr{P}$ in that case \textit{is} a subgroup of $\mathscr{B}$.\\
In Minkowski space-time a null hypersurface is said to be a \textit{good cone} if it is the future light cone 
of some point, and a \textit{bad cone} if its generators do not meet at a point. Consequently we define a 
\textit{good cross-section}, often called a \textit{good cut}, a cross-section of $\mathscr{I^+}$ which is 
the intersection of a future light cone of some point and the null hypersurface $\mathscr{I}^+$. 
A \textit{bad cut} is, on the other hand, the intersection of $\mathscr{I}^+$ with some null hypersurface 
which does not come together cleanly at a single vertex. The situation is represented in Figure \ref{fig:6.4}.
\begin{figure}[h]
\begin{center}
\includegraphics[scale=0.4]{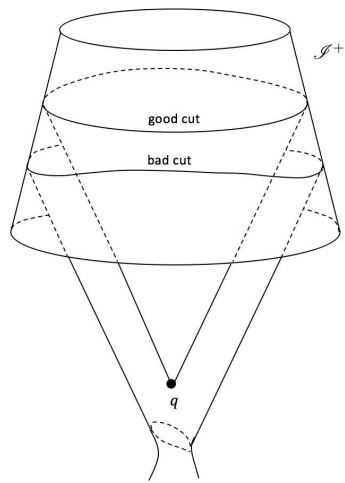}
\caption{A good cross-section of $\mathscr{I}^+$ for Minkowski space-time is one arising as the intersection 
of $\mathscr{I}^+$ with the future light cone of a point.}
\label{fig:6.4}
\end{center}
\end{figure}
\medskip \medskip \\ Using Bondi-Sachs coordinates the Minkowski metric tensor takes the form
\begin{equation*}
g=du\otimes du+du\otimes dr+dr\otimes du-r^2(d\theta\otimes d\theta+\sin^2\theta d\phi\otimes d\phi).
\end{equation*}
We see that each cut of $\mathscr{I}^+$ given by $u=t-r=\mathrm{const}$ is a good cut, since it arises from 
the future light cone of a point on the origin-axis $r=0$. In particular, all the good cuts can be obtained 
from the one given by $u=0$ by means of a space-time translation. Hence, as discussed in Sect.\ref{sect:6} we 
obtain that every good cut can be expressed in the form
\begin{equation*}
u=\frac{A+B\zeta+\bar{B}\bar{\zeta}+C\zeta\bar{\zeta}}{1+\zeta\bar{\zeta}}=\left(\frac{A+C}{2}\right)+\left(\frac{C-A}{2}\right)\cos\theta
\end{equation*}
\begin{equation}
\label{eqn:235}
+\left(\frac{B+\bar{B}}{2}\right)
\sin\theta\cos\phi+i\left(\frac{B-\bar{B}}{2}\right)\sin\theta\sin\phi,
\end{equation}
$A,C$ being real and $B$ being complex. Hence the equations describing good cuts are given, generally, 
by setting $u$ equal to a function of $\theta$ and $\phi$ which consists only of zeroth- 
and first-order spherical harmonics.\\
The effect of a transformation of the connected component of the Lorentz group is to leave invariant the 
particular good cut $u=0$. Such a transformation is
\begin{subequations}
\label{eqn:237}
\begin{align}
&\zeta\rightarrow\zeta'=\frac{a\zeta+b}{c\zeta+d},\\
&u\rightarrow u'=\frac{1+\zeta\bar{\zeta}}{|a\zeta+b|^2+|c\zeta+d|^2}u,
\end{align}
\end{subequations}
with $a,b,c,d\in\mathbb{C}$ such that $ad-bc=1$. Note that transformations \eqref{eqn:237} preserve the 
functional form of good cuts given in \eqref{eqn:235}. In fact we have, applying \eqref{eqn:237}, that
\begin{equation*}
u'=\left(\frac{A+B\zeta+\bar{B}\bar{\zeta}+C\zeta\bar{\zeta}}{1+\zeta\bar{\zeta}}\right)
\frac{1+\zeta\bar{\zeta}}{|a\zeta+b|^2+|c\zeta+d|^2},
\end{equation*}
where $\zeta$ and $\bar{\zeta}$ have now to be expressed as functions of $\zeta'$ and $\bar{\zeta}'$. 
It is straightforward to show that
\begin{equation*}
u'=\frac{A'+B'\zeta'+\bar{B}'\bar{\zeta}'+C'\zeta'\bar{\zeta}'}{1+\zeta'\bar{\zeta}'},
\end{equation*}
where 
\begin{align*}
&A'=A|a|^2-Bb\bar{a}-\bar{B}\bar{b}a+C|b|^2,\\
&B'=Bd\bar{a}+\bar{B}\bar{b}c-A\bar{a}c-Cd\bar{b}	,\\
&C'=A|c|^2-B\bar{c}d-\bar{B}c\bar{d}+C|d|^2.
\end{align*}
For example, if we perform a boost in the $z$ direction we have from \eqref{eqn:236} 
$a=e^{\chi/2},d=e^{-\chi/2},c=b=0$. Hence we get $A'=e^{\chi}A,B'=B,C'=e^{-\chi}C$. \\
Now it is clear that the general BMS transformation which maps good cuts into good cuts must obtain the particular 
good cut $u=0$ from some other good cut. We can therefore express the BMS transformation as the composition 
of a translation which maps this other good cut into $u=0$, with a Lorentz transformation which leaves 
$u=0$ invariant. Thus, the BMS transformation is
\begin{subequations}
\label{eqn:238}
\begin{align}
&\zeta\rightarrow\zeta'=\frac{a\zeta+b}{c\zeta+d},\\
&u\rightarrow u'=\left(\frac{1+\zeta\bar{\zeta}}{|a\zeta+b|^2+|c\zeta+d|^2}\right)\left(u+\frac{A+B\zeta
+\bar{B}\bar{\zeta}+C}{1+\zeta\bar{\zeta}}\right).
\end{align}
\end{subequations}
These BMS transformations form a $10$-real-parameters group. This is exactly the Poincar\'e group of Minkowski 
space-time, being a composition of a Lorentz transformation and a translation. We have obtained the following 
\begin{prop}$\\ $
The Poincar\'e group $\mathscr{P}$ is the group of transformations which maps good cuts into good 
cuts in Minkowski space-time.
\end{prop}
However, there are many other subgroups of the BMS group which can be expressed in the form \eqref{eqn:238} and 
which are therefore isomorphic with the Poincar\'e group. In fact $\mathscr{L}$ is not a normal subgroup of the 
BMS group since for any $b=(\Lambda,\alpha)\in\mathscr{B}$ and for any $\Lambda'=(\Lambda',0)\in\mathscr{L}$ 
the product $b\Lambda'b^{-1}$ is not necessarily an element of $\mathscr{L}$, as can be easily verified, and 
hence $\mathscr{L}$ does not get canonically singled out, occurring only as a factor group of $\mathscr{B}$ 
by the infinite-parameter Abelian group of supertranslations $\mathscr{S}$. In particular, Lorentz transformations 
do not commute with supertranslations and if we take any supertranslation $s$ and consider the group 
$\mathscr{L}'=s\mathscr{L}s^{-1}$ then it is a subgroup of the BMS group which is distinct from $\mathscr{L}$ 
but still isomorphic, and thus equivalent, to it. Explicitly, having fixed a supertranslation $s=(\mathbb{I},\alpha)$, 
a transformation of $\mathscr{L}'$ reads as
\begin{subequations}
\label{eqn:239}
\begin{align}
&\zeta\rightarrow\zeta'=\frac{a\zeta+b}{c\zeta+d},\\
&u\rightarrow u'=\left(\frac{1+\zeta\bar{\zeta}}{|a\zeta+b|^2+|c\zeta+d|^2}\right)\left(u-\alpha\right)+\alpha.
\end{align}
\end{subequations}
If we start with the good cut described by the equation $u=0$, which is left invariant by $\mathscr{L}$, and 
perform the supertranslation $s$, we obtain a new (bad) cut given by $u=\alpha$. This is the cut which is left 
invariant by \eqref{eqn:239} and hence by $\mathscr{L}'$. Hence $\mathscr{L}'$ maps bad cuts into bad cuts.
It follows that if we conjugate the whole Poincar\'e group $\mathscr{P}$ of \eqref{eqn:238} with respect to 
any supertranslation $s$ which is not a translation obtaining $\mathscr{P}'=s\mathscr{P}s^{-1}$ we get a 
distinct subgroup of $\mathscr{B}$, but isomorphic and completely equivalent to $\mathscr{P}$, which maps 
bad cuts into bad cuts. Of course, for a general $s$, $\mathscr{P'}$ and $\mathscr{P}$ have only the 
translations $\mathscr{T}$ in common. There exist \textit{many} subgroups of $\mathscr{B}$ which are 
isomorphic with $\mathscr{P}$ and hence the Poincar\'e group \textit{is not} a subgroup of the BMS group 
in a canonical way. However we have just seen that in Minkowski space-time, if we require the group of 
transformations to preserve the conformal nature of $\mathscr{I}^+$ and the strong conformal geometry, together 
with the property of mapping good cuts into good cuts, just one of the several copies of the 
Poincar\'e group gets singled out.
\begin{oss}$\\ $
Note that, again, the situation is similar to what happens for $\mathscr{L}$ within $\mathscr{P}$. In fact 
$\mathscr{L}$ does not arise naturally as a subgroup of $\mathscr{P}$, since if we form the group 
$\mathscr{L}'=t\mathscr{L}t^{-1}$, where $t$ is a translation, it is a different subgroup of $\mathscr{P}$ 
but isomorphic to $\mathscr{L}$. We can say, since the commutator of a Lorentz transformation and a 
translation is a translation, that $\mathscr{L}$, as a subgroup of $\mathscr{P}$, depends on the choice 
of an arbitrary origin in Minkowski space-time. 
\end{oss}
We turn now to the case when the space-time is asymptotically flat. The difficulty here is that there seems 
to be no suitable family of cuts that can properly take over the role of Minkowskian good cuts. This means 
that, although the translation elements of $\mathscr{B}$ are canonically singled out, there is no canonical 
concept of a \lq supertranslation-free\rq\hspace{0.1mm} Lorentz transformation. Hence the notion of a 
\lq pure translation\rq\hspace{0.1mm} still makes sense, but that of \lq pure rotation\rq\hspace{0.1mm} 
or \lq pure boost\rq\hspace{0.1mm} does not. However, as remarked by \cite{Lusanna2015,Lusanna2001} it can be shown that the use of appropriate \lq post-Minkowskian\rq\hspace{0.1mm} boundary conditions in the 3+1 makes it possible to get rid of the supertranslations and to single out the asymptotic Poincar\'e ADM group.\\
The most obvious generalization, for an asymptotically flat space-time, of the Minkowskian definition 
of a good cut, i.e. the intersection of future light cone of a point with $\mathscr{I}^+$, is totally 
inappropriate. One first reason is that there are many perfectly reasonable asymptotically flat space-times 
in which no cuts of $\mathscr{I}^+$ at all would arise in this way, e.g. \cite{Pen72}. Even if we 
restrict attention only to asymptotically flat space-times which do contain a reasonable number of good cuts 
of this kind, we are not likely to obtain any of the BMS transformations (apart from the identity) which 
maps this system of cuts into itself. The difficulty lies in the fact that the detailed irregularities of 
the interior of the space-time would be reflected in the definition of \lq goodness\rq\hspace{0.1mm} 
of a cut. In other words, the light-cone cuts are far more complicated than those \eqref{eqn:235} of flat space. \\
However, there is a more satisfactory way to characterize good cuts of $\mathscr{I}^+$, based on the shear 
of null hypersurfaces intersecting $\mathscr{I}^+$. Suppose now that the (physical) space-time under consideration 
contains a null curve $\mu$ of a null geodesic congruence $\mathscr{C}$ affinely parametrized by $\tilde{r}$. 
It can be shown \cite{Penrin2} that the physical shear of the null hypersurfaces generated by $\mathscr{C}$ 
has the following asymptotic behaviour for large values of $\tilde{r}$:
\begin{equation*}
\tilde{\sigma}=\frac{\tilde{\sigma}^0}{\tilde{r}^2}+O(\tilde{r}^{-4}),
\end{equation*}
where $\tilde{\sigma}^0$ is called \textit{asymptotic shear}. In the case of a flat space-time the vanishing of 
the asymptotic shear implies the vanishing of the whole shear, but if we turn to the case of asymptotically 
flat space-times that may contain matter, although the leading term of the asymptotic behaviour of 
$\tilde{\sigma}$ does not change, the vanishing of $\tilde{\sigma}^0$ does not imply the vanishing of 
$\tilde{\sigma}$ \cite{New2012}.\\
It can be shown \cite{GRG} that if we consider the unphysical space-time obtained with conformal factor 
$\Omega=\tilde{r}^{-1}$, the shear transforms as
\begin{equation*}
\sigma=\Omega^{-2}\tilde{\sigma},
\end{equation*}
and hence
\begin{equation*}
\left.\sigma\right|_{_{\mathscr{I}^+}}=\tilde{\sigma}^0.
\end{equation*}
In Minkowski space-time the good cones are characterized locally by the fact that the null rays generating them 
possess no shear and it can be shown that the cuts of $\mathscr{I}^+$ which we defined earlier are precisely the 
ones arising from the intersection of $\mathscr{I}^+$ with null hypersurfaces characterized by $\tilde{\sigma}^0=0$. 
Thus, a definition of \lq goodness\rq\hspace{0.1mm} is provided, for Minkowski space-time, which refers only 
to quantities defined asymptotically. \\
We would like to extend this definition of good cut to asymptotically flat space-times too. We could say that 
some cut is a good cut if its complex shear equals zero (since on $\mathscr{I}^+$ we have $\sigma=\tilde{\sigma}^0$). We cannot, however, define good cones simply by requiring $\tilde{\sigma}^0=0$. In many cases it is not possible 
to arrange $\tilde{\sigma}^0=0$ for all values of $\theta$ and $\phi$. But even in cases where it is possible 
we have another problem, which is due to the presence of gravitational radiation. To make this point clear, 
we cite now some important results regarding the relation between asymptotic shear and gravitational radiation 
which are basically due to \cite{Bondi62,Sachs62,NewUn}. A first result is that $\tilde{\sigma}^0$ forms part 
of the initial data on $u=0$ used to determine the space-time asymptotically. Furthermore, it turns out that
\begin{equation*}
\frac{\partial \tilde{\sigma}^0}{\partial u}=-\bar{N}
\end{equation*}
where $N$ is the \textit{Bondi news function} and that the rate of energy-momentum loss due to gravitational 
radiation through a hypersurface $\mathscr{S}$ which spans some two-dimensional cross-section 
$S$ of $\mathscr{I}^+$ is
\begin{equation}
\label{eqn:244}
\frac{dP^a}{du}=-\frac{1}{4\pi}\int W^a|N|^2dS,
\end{equation}
where 
\begin{equation*}
W^0=1,\hspace{0.5cm}W^1=\sin\theta\cos\phi,\hspace{0.5cm}W^2=\sin\theta\sin\phi,\hspace{0.5cm}W^3=\cos\theta.
\end{equation*}
\begin{figure}[h]
\begin{center}
\includegraphics[scale=0.37]{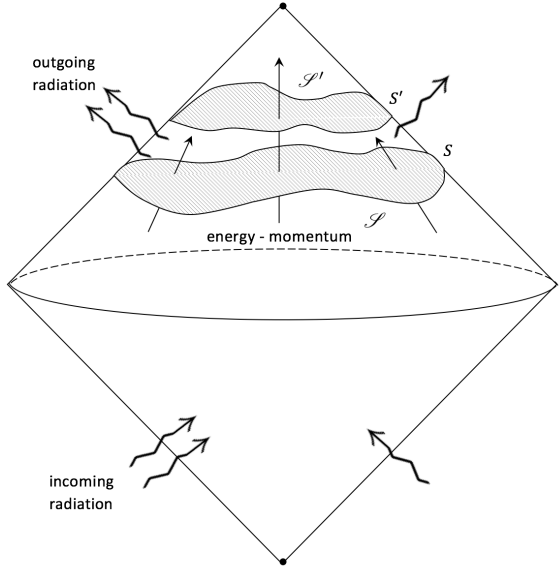}
\caption{The outgoing gravitational radiation through two hypersurfaces $\mathscr{S}$ and $\mathscr{S}'$ with 
associated cross-sections $S$ and $S'$ respectively.}
\label{fig:6.5}
\end{center}
\end{figure}
\\Note that $S$ is topologically a sphere $S^2$ and can always be transformed, by the introduction of a suitable 
conformal factor, into a metric sphere of unit radius. Hence $dS$ can be taken to be	
\begin{equation*}
dS=\sin\theta d\theta d\phi.
\end{equation*}
Thus, the squared modulus of $N$ represents the flux of energy-momentum of the outgoing gravitational radiation. 
\\The time component of \eqref{eqn:244} gives the famous \textit{Bondi-Sachs mass-loss formula}:
\begin{equation}
\label{eqn:243}
\frac{dm_{_{B}}}{du}=-\frac{1}{4\pi}\int |N|^2dS\leq 0.
\end{equation}
An useful representation of the situation is furnished by Figure \ref{fig:6.5}. The positivity of the integrand in \eqref{eqn:243} shows that if a system emits gravitational waves, i.e. if 
there is news, then its Bondi mass must decrease. If there is no news, i.e. $N=0$, the Bondi mass is constant. 
The reason we get a mass loss rather than a mass gain is simply that all we have said has been applied 
to $\mathscr{I}^+$ instead of $\mathscr{I}^-$.\\
From the above discussion it follows that, if the cut $S$ is given by $u=0$ and is shear-free, the cross-sections 
$u=\mathrm{const}$, which are translations of $\mathscr{I}^+$, will not be shear-free in the presence of 
gravitational radiation. In other words, if $\tilde{\sigma}^0=0$ for one value of $u$, we will generally 
have $\tilde{\sigma}^0\neq 0$ for a later value of $u$, i.e. \lq goodness\rq\hspace{0.1mm} would not be 
invariant under translation. However, a difficulty arises even more serious than this. Since to specify 
a cut we just need to specify the value of $u$ on each generator of $\mathscr{I}^+$, the freedom in choosing 
a cut is one real number per point of the cut. On the other hand the quantity $\tilde{\sigma}^0$ is 
complex, its vanishing therefore, representing two real numbers per point of the section. We briefly 
discuss, without going into details \cite{NewPen66} how to solve this problem. The first step is to define 
the magnetic and the electric part of $\tilde{\sigma}^0$:
\begin{equation*}
\tilde{\sigma}^0=\tilde{\sigma}^0_e+\tilde{\sigma}^0_m
\end{equation*}
It can be shown that for each hypersurface $u=\mathrm{const}$ the splitting of $\tilde{\sigma}^0$ into 
its electric and magnetic part is invariant under BMS conformal transformation, i.e. Lorentz transformations 
which leave the surface $u=\mathrm{0}$ invariant. Furthermore, in the Minkowski case, the asymptotic 
shear $\tilde{\sigma}^0=\tilde{\sigma}^0_e+\tilde{\sigma}^0_m$ behaves, when $u\rightarrow-\infty$ as
\begin{subequations}
\label{eqn:250}
\begin{align}
&\tilde{\sigma}^{0}_e(u,\theta,\phi)\xrightarrow[u \rightarrow -\infty]{}S_e(\theta,\phi),\\
&\tilde{\sigma}^{0}_m(u,\theta,\phi)\xrightarrow[u \rightarrow -\infty]{}0,
\end{align}
\end{subequations}
with $S_e(\theta,\phi)$ purely electric and independent of $u$. The magnetic part $\tilde{\sigma}^0_m(u,\theta,\phi)$, as 
$u\rightarrow-\infty$ vanishes, i.e. $S_m(\theta,\phi)=0$.\\
On the basis of what happens in the Minkowski case we wish to impose a physical restriction on the behaviour of 
$\tilde{\sigma}^0$ as $u\rightarrow-\infty$ for a generic space-time. Although no actual cuts of $\mathscr{I}^+$ 
may be shear-free, it is reasonable to expect that in the limit $u\rightarrow-\infty$ on $\mathscr{I}^+$, 
such cuts will exist. Requiring that this limiting shear-free cuts be mapped into one another, we can actually 
restrict the BMS transformations to obtain a canonically defined subgroup of the BMS group, which is isomorphic 
to the Poincar\'e group. The Poincar\'e group which emerges in this way, by virtue of the considerations we have 
developed on the gravitational radiation, may be thought of as that which has relevance to the remote past, 
\textit{before} all the gravitational radiation has been emitted. In analogy with \eqref{eqn:250} we require that
\begin{equation}
\label{eqn:248}
\tilde{\sigma}^0(u,\theta,\phi)\longrightarrow S(\theta,\phi).
\end{equation}
If the analogy with the Minkowski theory can be trusted, we would expect $S(\theta,\phi)$ to be purely electric. 
However, it is not essential since it will be possible to extract the Poincar\'e group only on the basis of 
\eqref{eqn:248}. Hence we will treat the case in which $S(\theta,\phi)$ could have a magnetic part too, i.e. 
$S(\theta,\phi)=S_e(\theta,\phi)+S_m(\theta,\phi)$ with $S_m(\theta,\phi)\neq 0$. \\
It can be shown \cite{Sachs62} that, under a conformal BMS transformation, the asymptotic shear transforms as
\begin{equation}	
\label{eqn:249}
\tilde{\sigma}^0{}'(u,\theta,\phi)=K^{-1}\left[\tilde{\sigma}^0(u,\theta,\phi)+\frac{1}{2}\eth^2\alpha(\theta,\phi)\right],
\end{equation}
where $\eth$ is a differential operator, whose properties can be found in \cite{New67}. It has to be remarked 
that $\tilde{\sigma}^0{}'(u,\theta,\phi)$ refers to the asymptotic shear of the hypersurface $u'=\mathrm{const}$ 
of the transformed coordinate system evaluated at $(u,\theta,\phi)$. The complete transformation 
$\tilde{\sigma}^0{}'(u',\theta',\phi')$ is more complicated. Applying \eqref{eqn:248} to \eqref{eqn:249} gives
\begin{align*}
&S'_e(\theta,\phi)=K^{-1}\left[S_e(\theta,\phi)-\frac{1}{2}\eth^2\alpha(\theta,\phi)\right],\\
&S'_m(\theta,\phi)=K^{-1}S_m(\theta,\phi),
\end{align*}
since $\alpha(\theta,\phi)$ is real and the magnetic part is imaginary. It is always possible to set
\begin{equation*}
S_e(\theta,\phi)=\eth^2G(\theta,\phi),
\end{equation*}
for some real $G(\theta,\phi)$. Since $\alpha(\theta,\phi)$ can be chosen arbitrarily on the sphere, it follows 
that a BMS transformations for which $\alpha(\theta,\phi)=2G(\theta,\phi)$ imposes $S'_e(\theta,\phi)=0$. 
Thus, we have introduced coordinate conditions for which $S_e(\theta,\phi)=0$ at $u=-\infty$. Now the BMS 
transformations which preserve the condition $S_e(\theta,\phi)=0$ are those for which
\begin{equation*}
\eth^2\alpha(\theta,\phi)=0.
\end{equation*}
It can be easily shown that this condition restricts $\alpha(\theta,\phi)$ to be of the form \eqref{eqn:207} 
and thus the allowed supertranslations are simply the translations. The Lorentz transformations, given by 
$\alpha(\theta,\phi)=0$ do not spoil the coordinate conventions. We have finally obtained the following: 
\begin{thm}$\\ $
The group of asymptotic isometries of an asymptotically flat space-time which preserves the condition 
$S_e(\theta,\phi)=0$ at $u=-\infty$ is isomorphic with the Poincar\'e group $\mathscr{P}$.
\end{thm}
\begin{oss}
We could have carried out the same arguments by taking the limit $u\rightarrow +\infty$, and it would 
have been an independent choice. Thus, in a similar way, we could have extracted another Poincar\'e 
group which has relevance to the remote future, i.e. \textit{after} all the gravitational radiation has 
been emitted. There seems to be no reason to believe that these two Poincar\'e groups will be the same, in general. 
\end{oss}
\section{BMS transformations and gravitational scattering}
\label{sect:9}
The ground is now ready for considering a very recent application of BMS transformations,
i.e. the discovery by Strominger \cite{ST2014} that there exist BMS transformations acting non-trivially on
outgoing gravitational scattering data while preserving the intrinsic structure at future null
infinity. His analysis begins with the local expression of a generic Lorentzian metric in
retarded Bondi coordinates, that we know from Sect. \ref{sect:2}, and in particular
with the asymptotic expansion of the metric about future null infinity (where $r=\infty$),
reading as (cf. \eqref{eqn:140})
\begin{eqnarray}
g &= &du \otimes du +(du \otimes dr + dr \otimes du)
-r^{2}\gamma_{z {\bar z}}(dz \otimes d{\bar z}+d {\bar z} \otimes dz)
 \nonumber \\
&- &2{m_{B}^+\over r}du \otimes du - r C_{zz}dz \otimes dz 
-r C_{{\bar z}{\bar z}}d{\bar z}\otimes d{\bar z} 
\nonumber \\ 
&+ &U_{z}(dz \otimes d{\bar z}+d{\bar z} \otimes dz)
+U_{{\bar z}}(du \otimes d{\bar z}+d{\bar z} \otimes du)
+...,
\label{(9.1)}
\end{eqnarray}
where 
\begin{equation}
\gamma_{z{\bar z}} \equiv {2 \over (1+z{\bar z})}, \;
U_{z} \equiv -{1 \over 2}D^{z}C_{zz},
\label{(9.2)}
\end{equation}
$D^{z}$ denoting covariant differentiation with respect to the $2$-sphere metric 
$\gamma_{z{\bar z}}$. The vector fields that 
generate BMS transformations are of two types: there are $6$ with an
asymptotic $\mathrm{SL}(2,\mathbb{C})$ Lie bracket algebra, and an infinite number of commuting supertranslations.
Global $\mathrm{SL}(2,\mathbb{C})$ conformal transformations are generated on $\mathscr{I}^{+}$ by the real part of
the complex vector fields (cf. Sect. \ref{sect:7})
\begin{eqnarray}
{}^{(+)}\zeta^{a}\partial_{a}&=&\left(1-{u \over 2r}\right){}^{(+)}\zeta^{z}\partial_{z}
-\left(1+{u \over r}\right){D_{z}{}^{(+)}\zeta^{z}r \over 2}\partial_{r}
\nonumber \\&-&{u \over 2r}\gamma^{z {\bar z}}D_{z}^{2}{}^{(+)}\zeta^{z}\partial_{{\bar z}}
+{u \over 2}D_{z}{}^{(+)}\zeta^{z}\partial_{u}
+{\rm O}(r^{-2}),
\label{(9.3)}
\end{eqnarray}
where ${}^{(+)}\zeta^{z}=(1,z,z^{2},i,iz,iz^{2})$ \cite{ST2014}. The action
of the Lie derivative operator yields \cite{ST2014}
\begin{align}
&{\cal L}_{{}^{(+)}\zeta}C_{zz}={}^{(+)}\zeta^{z}\partial_{z}C_{zz}
+2 \partial_{z}{}^{(+)}\zeta^{z}C_{zz}+{D_{z}{}^{(+)}\zeta^{z}\over 2}
(u \partial_{u}-1)C_{zz},
\label{(9.4)} \\
&{\cal L}_{{{}^{(+)}\bar \zeta}}C_{zz}={{}^{(+)}\bar \zeta}^{{\bar z}}\partial_{{\bar z}}C_{zz}
+{D_{{\bar z}}{{}^{(+)}\bar \zeta}^{{\bar z}}\over 2}(u \partial_{u}-1)C_{zz},
\label{(9.5)} \\
&{\cal L}_{{}^{(+)}\zeta}m_B^+=\left({}^{(+)}\zeta^{z}\partial_{z}+{u \over 2}D_{z}{}^{(+)}\zeta^{z}\partial_{u}
+{3 \over 2}D_{z}{}^{(+)}\zeta^{z}\right)m_{B}^+ \nonumber \\+&{u \over 2}\partial_{u}
(U_{z}{}^{(+)}\zeta^{z}-U^{z}D_{z}^{2}{}^{(+)}\zeta^{z}).
\label{(9.6)}
\end{align}

The $\mathrm{SL}(2,\mathbb{C})$ transformations on $\mathscr{I}^{-}$ are generated by the complex vector 
fields (cf. (9.3) and replace $u$ by $-v$ therein)
\begin{eqnarray}
{}^{(-)}\zeta^{a}\partial_{a} &=&\left(1+{v \over 2r}\right){}^{(-)}\zeta^{z}\partial_{z}
-\left(1-{v \over r}\right){D_{z}{}^{(-)}\zeta^{z}r \over 2}\partial_{r}
\nonumber \\&+&{v \over 2r}\gamma^{z {\bar z}}D_{z}^{2}{}^{(-)}\zeta^{z}\partial_{{\bar z}}
+{v \over 2}D_{z}{}^{(-)}\zeta^{z}\partial_{v}
+{\rm O}(r^{-2}),
\label{(9.7)}
\end{eqnarray}
and the associated Lie derivative acts according to
\begin{align}
&{\cal L}_{{}^{(-)}\zeta}D_{zz}={}^{(-)}\zeta^{z}\partial_{z}D_{zz}
+2 \partial_{z}{}^{(-)}\zeta^{z}D_{zz}+{D_{z}{}^{(-)}\zeta^{z}\over 2}
(v \partial_{v}-1)D_{zz},
\label{(9.8)} \\
&{\cal L}_{{{}^{(-)}\bar \zeta}}D_{zz}={{{}^{(-)}\bar \zeta}}^{{\bar z}}\partial_{{\bar z}}D_{zz}
+{D_{{\bar z}}{{}^{(-)}\bar \zeta}^{{\bar z}}\over 2}(v \partial_{v}-1)D_{zz},
\label{(9.9)} \\
&{\cal L}_{{}^{(-)}\zeta}m_{B}^{-}=\left({}^{(-)}\zeta^{z}\partial_{z}+{v \over 2}D_{z}{}^{(-)}\zeta^{z}\partial_{v}
+{3 \over 2}D_{z}{}^{(-)}\zeta^{z}\right)m_{B}^{-} \nonumber \\&-{v \over 2}\partial_{v}
(V_{z}{}^{(-)}\zeta^{z}-V^{z}D_{z}^{2}{}^{(-)}\zeta^{z}).
\label{(9.10)}
\end{align}
The supertranslations on $\mathscr{I}^{+}$ are generated by the vector fields
\begin{equation}
{}^{(+)}f \partial_{u}-{1 \over r}\Bigr(D^{{\bar z}}{}^{(+)}f \partial_{{\bar z}}
+D^{z}{}^{(+)}f \partial_{z}\Bigr)+D^{z}D_{z}{}^{(+)}f \partial_{r}, \;
{}^{(+)}f={}^{(+)}f(z,{\bar z}),
\label{(9.11)}
\end{equation}
and one finds the Lie derivatives
\begin{align}
&{\cal L}_{{}^{(+)}f}C_{zz}={}^{(+)}f \partial_{u}C_{zz}-2D_{z}^{2}{}^{(+)}f,
\label{(9.12)}\\
&{\cal L}_{{}^{(+)}f}U_{z}={}^{(+)}f \partial_{u}U_{z}
-{1 \over 2}D^{z}{}^{(+)}f \partial_{u}C_{zz}+D^{z}D_{z}^{2}{}^{(+)}f.
\label{(9.13)}
\end{align}

Supertranslations on past null infinity $\mathscr{I}^{-}$ are also of interest, and are
generated by the vector fields
\begin{equation}
{}^{(-)}f\partial_{v}+{1 \over r}\Bigr(D^{{\bar z}}{}^{(-)}f\partial_{{\bar z}}
+D^{z}{}^{(-)}f\partial_{z} \Bigr)-D^{z}D_{z}{}^{(-)}f\partial_{r},
\label{(9.14)}
\end{equation}
the Lie derivatives along which read as \cite{ST2014}
\begin{align}
&{\cal L}_{{}^{(-)}f}D_{zz}={}^{(-)}f\partial_{v}D_{zz}+2D_{z}^{2}{}^{(-)}f,
\label{(9.15)}\\
&{\cal L}_{{}^{(-)}f}V_{z}={}^{(-)}f\partial_{u}V_{z}+{1 \over 2}D^{z}{}^{(-)}f\partial_{u} D_{zz}
+D^{z}D_{z}^{2}{}^{(-)}f.
\label{(9.16)}
\end{align}
We are now going to outline the connection between BMS transformations on
$\mathscr{I}^{+}$ and $\mathscr{I}^{-}$, denoted by ${\rm BMS}^{+}$ and
${\rm BMS}^{-}$, respectively.

\subsection{The Christodoulou-Klainerman space-times}

After the singularity theorems of Penrose, Hawking and Geroch 
\cite{CS1,CS2,CS3,CS4,CS5,CS6,CS7,CS8,CS9,CS10,CS11}, it was thought for a long
time that singularities are a generic property of general relativity, if the energy-momentum 
tensor obeys suitable conditions. However, an outstanding piece of work of Christodoulou and
Klainerman \cite{CK1993} proved that there exist asymptotically flat space-times that are
geodesically complete and hence singularity-free in classical theory.

Christodoulou and Klainerman studied asymptotically flat initial data in the center-of-mass 
frame on a maximal spacelike slice for which the Bach tensor
$$
\varepsilon^{ijk} \; { }^{(3)}D_{j} \; { }^{(3)}G_{kl}
$$
of the induced $3$-metric decays at least as $r^{-{7 \over 2}}$ at spatial infinity,
while the extrinsic curvature decays like $r^{-{5 \over 2}}$. This implies in particular
that, in normal coordinates about infinity, the leading part of the $3$-metric takes the
Schwarzschild form, while the correction decays like $r^{-{3 \over 2}}$. Remarkably, all such
initial data, when supplemented by a global smallness condition, give rise to a geodesically
complete solution of the Einstein equations. Since such a smallness condition holds in a
finite neighbourhood of Minkowski space-time, the result obtained is said to be the global
non-linear stability of Minkowski space-time \cite{CK1993}.
 
As we know from Sect. \ref{sect:8}, the gravitational radiation flux is proportional to the square of
the Bondi news $N_{zz}$ which, for any finite energy data, vanishes on the boundaries
$\partial_{+}\mathscr{I}^{+}$ and $\partial_{-}\mathscr{I}^{+}$ of future null infinity, i.e.
\begin{equation}
\left . N_{zz} \right|_{\partial_{\pm} \mathscr{I}^{+}}=0.
\label{(9.17)}
\end{equation}
For Christodoulou-Klainerman space-times the falloff property for $u \rightarrow \pm \infty$ is
\begin{equation}
N_{zz}(u) \sim |u|^{-{3 \over 2}},
\label{(9.18)}
\end{equation}
or even faster than this. With the coordinates used in \eqref{(9.1)}, the Weyl curvature component
$\Psi_{2}^{0}$ reads as
\begin{eqnarray}
\Psi_{2}^{0}(u,z,{\bar z})&=&-\lim_{r \to \infty} \Bigr(r C_{uzr{\bar z}} \; \gamma^{z{\bar z}}\Bigr) \nonumber \\
&=&-m_{B}^++{1 \over 4}C^{zz}N_{zz}
-{1 \over 2}\gamma^{z {\bar z}}\Bigr(\partial_{\bar z}U_{z}
-\partial_{z}U_{\bar z}\Bigr).
\label{(9.19)}
\end{eqnarray}
For center-of-mass Christodoulou-Klainerman spaces at $r=\infty,u=\infty$ one finds
\begin{equation}
\left . \Psi_{2}^{0}\right |_{\partial_{+}\mathscr{I}^{+}}=0,
\label{(9.20)}
\end{equation}
while at $r=\infty,u=-\infty$
\begin{equation}
\left . \Psi_{2}^{0}\right |_{\partial_{-}\mathscr{I}^{+}}=-M,
\label{(9.21)}
\end{equation}
the ratio ${M \over G}$ being the Arnowitt-Deser-Misner mass. The real part of these formulae, 
jointly with \eqref{(9.17)}, imply for the mass aspect the following boundary conditions:
\begin{align}
&\left . m_{B}^+ \right |_{\partial_{+}\mathscr{I}^{+}}=0,
\label{(9.22)}\\
&\left . m_{B}^+ \right |_{\partial_{-}\mathscr{I}^{+}}=M.
\label{(9.23)}
\end{align}
From the imaginary part of \eqref{(9.19)} one obtains 
\begin{equation}
\Bigr[\partial{{\bar z}}U_{z}-\partial_{z}U_{{\bar z}}\Bigr]_{\partial_{\pm}\mathscr{I}^{+}}=0.
\label{(9.24)}
\end{equation}
A similar procedure on past null infinity yields the boundary conditions
\begin{equation}
\left . m_{B}^{-} \right |_{\partial_{-}\mathscr{I}^{-}}
=\Bigr[\partial_{{\bar z}}V_{z}-\partial_{z}V_{{\bar z}}\Bigr]_{\partial_{\pm}\mathscr{I}^{-}}
= \left . \partial_{v}D_{zz} \right|_{\partial_{\pm}\mathscr{I}^{-}}=0,
\label{(9.25)}
\end{equation}
and
\begin{equation}
\left . m_{B}^{-} \right |_{\partial_{+}\mathscr{I}^{-}}=M.
\label{(9.26)}
\end{equation}
The falloff of the Bondi news, here denoted by $M_{zz}$, is
\begin{equation}
M_{zz} \equiv \partial_{v}D_{zz} \sim |v|^{-{3 \over 2}},
\label{(9.27)}
\end{equation}
or even faster.

Note that, while it seems reasonable that the total integrated mass should decay to zero
at the future boundary of future null infinity for a sufficiently weak gravitational disturbance, 
it is not a priori clear that the unintegrated mass aspect function $m_{B}$ should itself
approach $0$. Furthermore, the news tensor might decay more slowly near the future and past
boundaries of null infinity and still provides a finite value for the total energy flux
\cite{ST2014}. Strominger has also considered the possible coupling to any kind of massless
matter that dissipates at late times on $\mathscr{I}^{+}$ or early times on
$\mathscr{I}^{-}$. He imposes the boundary conditions \eqref{(9.20)}-\eqref{(9.26)} as restriction on
the desired solutions of Einstein's equations.

\subsection{Link among future and past null infinity near spacelike infinity}

The future null infinity data $m_{B}^+$ and $C_{zz}$ are related by the constraint \cite{ST2014}
\begin{equation}
\partial_{u}m_{B}^+=-{1 \over 2}\partial_{u}\Bigr[D^{z}U_{z}+D^{\bar z}U_{\bar z}\Bigr]-T_{uu},
\label{(9.28)}
\end{equation}
where $T_{uu}$ is the total outgoing radiation energy flux for gravity plus matter at 
future null infinity, whose explicit form is \cite{ST2014}
\begin{equation}
T_{uu}={1 \over 4}N_{zz}N^{zz}+4 \pi G \lim_{r \to \infty}
\Bigr[r^{2}T_{uu}^{M}\Bigr].
\label{(9.29)}
\end{equation}

A classical gravitational scattering problem consists of a set of final data on future
null infinity which evolve from some initial data on past null infinity. On going from a
Cristodoulou-Klainerman space-time to these data one encounters a physical ambiguity
under ${\rm BMS}^{+} \times {\rm BMS}^{-}$ transformations. In particular one has to
evaluate, once the Bondi news and matter radiation flux are given, the leading terms $m_{B}$
and $C_{zz}$ of the metric in the coordinates used in \eqref{(9.1)} everywhere on future null infinity.
In order to achieve this one integrates equation \eqref{(9.28)} and the equation
\begin{equation}
\partial_{u}C_{zz}=N_{zz}
\label{(9.30)}
\end{equation}
along a null generator of future null infinity, but initial conditions at the past
boundary of $\mathscr{I}^{+}$ are required. For Christodoulou-Klainerman space-times one
has indeed the initial condition \eqref{(9.23)}. Under a $\mathrm{SL}(2,\mathbb{C})$ transformation on future null
infinity, which forces us to abandon the center-of-mass frame, the restriction of $m_{B}$
to $\partial_{-}\mathscr{I}^{+}$ transforms according to \cite{ST2014}
\begin{equation}
{\cal L}_{{}^{(+)}{\zeta}} \left . m_{B}^+ \right |_{\partial_{-}\mathscr{I}^{+}}
=\left({}^{(+)}\zeta^{z}\partial_{z}+{3 \over 2}D_{z}{}^{(+)}\zeta^{z}\right)
\left . m_{B}^+ \right |_{\partial_{-}\mathscr{I}^{+}},
\label{(9.31)}
\end{equation}
and hence is no longer constant.

As far as the initial value of $C_{zz}$ is concerned, this may be any solution of
equation \eqref{(9.24)}, which can be re-expressed in the form
\begin{equation}
\Bigr[D_{\bar z}^{2}C_{zz}-D_{z}^{2}C_{{\bar z}{\bar z}}\Bigr]_{\partial_{-}\mathscr{I}^{+}}=0.
\label{(9.32)}
\end{equation}
If $C$ denotes some real-valued function, the general integral of equation \eqref{(9.32)} reads as
\begin{equation}
C_{zz}(-\infty,z,{\bar z})=D_{z}^{2}C.
\label{(9.33)}
\end{equation}
This implies that
\begin{equation}
U_{z}=-{1 \over 2}\partial_{z}(C+D_{z}D^{z}C).
\label{(9.34)}
\end{equation}

By proceeding along similar lines on past null infinity, the mass aspect $m_{B}^{-}$ is
determined by the constraint \cite{ST2014} 
\begin{equation}
\partial_{v}m_{B}^{-}={1 \over 2}\partial_{v}
\Bigr[D^{z}V_{z}+D^{\bar z}V_{\bar z}\Bigr]+T_{vv},
\label{(9.35)}
\end{equation}
where $T_{vv}$ denotes the total incoming radiation flux at past null infinity. This can be solved
by imposing data at $\partial_{+}\mathscr{I}^{-}$ and integrating until 
$\partial_{-}\mathscr{I}^{-}$ is reached. One can assume the boundary condition \eqref{(9.26)}. The
mass aspect has the transformation property
\begin{equation}
\left . {\cal L}_{{}^{(-)}\zeta}m_{B}^{-} \right |_{\partial_{+}\mathscr{I}^{-}}
=\left({}^{(-)}\zeta^{z}\partial_{z}+{3 \over 2}D_{z}{}^{(-)}\zeta^{z}\right)
\left . m_{B}^{-} \right |_{\partial_{+}\mathscr{I}^{-}} .
\label{(9.36)}
\end{equation}
Moreover, for some real-valued function $F$, one can take
\begin{equation}
D_{zz}(-\infty,z,{\bar z})=D_{z}^{2}F.
\label{(9.37)}
\end{equation}
The next task is to relate the initial data at $\partial_{+}\mathscr{I}^{-}$ to those
at $\partial_{-}\mathscr{I}^{+}$. For this purpose, one has first to relate the points
on $\partial_{+}\mathscr{I}^{-}$ to the points on $\partial_{-}\mathscr{I}^{+}$. 
The desired relation derives from the conformal infinity picture in which spacelike
infinity is a point and future (resp. past) null infinity is its future (resp. past) lightcone.
Null generators of null infinity then go from past null infinity to future null infinity
through spacelike infinity. These generators provide the necessary identification: all points
on the same generator are labeled by the same coordinate $(z,{\bar z})$.

The ordinary Lorentz transformations of Minkowski space-time are generated by vector fields
of the form \eqref{(9.3)} and \eqref{(9.7)}, restricted by the condition
\begin{equation}
{}^{(+)}\zeta^{z}={}^{(-)}\zeta^{z}.
\label{(9.38)}
\end{equation}
This condition is defined only after having identified points on
$\partial_{+}\mathscr{I}^{-}$ and $\partial_{-}\mathscr{I}^{+}$. Such Lorentz transformations 
preserve the condition
\begin{equation}
\left . m_{B}^+ \right |_{\partial_{-}\mathscr{I}^{+}}(z,{\bar z})
=\left . m_{B}^+ \right |_{\partial_{+}\mathscr{I}^{-}}(z,{\bar z}),
\label{(9.39)}
\end{equation}
despite the fact that both mass aspects change under boosts.

Furthermore, one has to relate the initial data for $C_{zz}$ at $\partial_{-}\mathscr{I}^{+}$
to those for $D_{zz}$ at $\partial_{+}\mathscr{I}^{-}$ in a way compatible with \eqref{(9.38)}. 
This is obtained by imposing
\begin{equation}
D_{zz}(\infty,z,{\bar z})=-C_{zz}(-\infty,z,{\bar z}).
\label{(9.40)}
\end{equation}
The equations for $U_{z}$ (see \eqref{(9.2)}) and $V_{z}$ then imply
\begin{equation}
\left . V_{z} \right |_{\partial_{+}\mathscr{I}^{-}}
=\left . U_{z} \right |_{\partial_{-}\mathscr{I}^{+}}.
\label{(9.41)}
\end{equation}
This simple but non-trivial condition unifies the separate BMS symmetries on future and past
null infinity into a single one acting on both parts of null infinity, so that
\begin{equation}
{}^{(-)}f(z,{\bar z})={}^{(+)}f(z,{\bar z}).
\label{(9.42)}
\end{equation}
If ${}^{(+)}f$ and ${}^{(-)}f$, besides being equal, are constant, it means we deal with global time
translations. The three constant spatial translations are given by the formula
\begin{equation}
f={(1-z{\bar z})\over (1+z{\bar z})}.
\label{(9.43)}
\end{equation}
Generators that obey \eqref{(9.38)} and \eqref{(9.42)} are the diagonal generators of 
${\rm BMS}^{+} \times {\rm BMS}^{-}$. Strominger has stressed \cite{ST2014} that the
possibility of linking data on future and past null infinity depends crucially on the
property of Christodoulou-Klainerman space-times, and such a link is not always possible
for more general types of asymptotically flat space-times. For example, if the news tensor
obeys the ${1 \over |u|}$ fall off, $C_{zz}$ diverges logarithmically, and the boundary
condition \eqref{(9.40)} cannot be imposed.

\subsection{Symmetries of classical scattering}

To sum up, one starts from a Christodoulou-Klainerman space-time in the center-of-mass frame,
jointly with conformally invariant radiative data $N_{zz}$ and $M_{zz}$ on null infinity.
In order to be able to regard this as a scattering solution, one has to construct both the
initial data on past null infinity and the final data on future null infinity. For this purpose,
one assumes that (9.39) is fulfilled with a boundary value $M$, while
\begin{equation}
\left . C_{zz} \right |_{\partial_{-}\mathscr{I}^{+}}(z,{\bar z})
=\left . D_{zz} \right |_{\partial_{+}\mathscr{I}^{-}}(z,{\bar z})=0,
\label{(9.44)}
\end{equation}
and one performs integration of the constraint equations on future and past null infinity
away from spacelike infinity. Hence one obtains final data $(C_{zz},m_{B}^+)$ that solve the
scattering problem in correspondence to the initial data $(D_{zz},m_{B}^{-})$.

Strominger has also stressed that, for every Christodoulou-Klainerman space, one can
generate an infinite-parameter family of solutions to the scattering problem. Some of these
differ by Poincar\'e transformations, while others differ by supertranslations.

\section{Recent developments and concluding remarks}

The latest applications of all concepts and properties described so far have to do with
black-hole physics in quantum gravity. Indeed, supertranslations transform the Minkowski
vacuum into a physically inequivalent zero-energy vacuum. This means that the vacuum is not
invariant, and hence supertranslation symmetry is spontaneously broken. But if the 
quantum gravity vacuum is not unique, it becomes conceivable that the final vacuum state
might get correlated with the thermal Hawking radiation in such a way that no loss of
unitarity ever occurs \cite{HPS1}. 

Moreover, Bondi-Metzner-Sachs transformations are 
diffeomorphisms that change the physical state (e.g. a supertranslation maps a stationary
black hole to a physically inequivalent one). During black hole evaporation, supertranslation
charge is radiated through null infinity, and since this charge is conserved, the sum of
the black hole and radiated supertranslation charge is fixed at all times. This implies
that black holes carry soft hair (i.e. additional data, besides mass, charge and
angular momentum), and if the black hole evaporates completely, the net supertranslation
charge in the outgoing radiation must be conserved. This leads in turn to correlations 
between the early- and late-time Hawking radiation \cite{HPS1}.  

In subsequent work \cite{HPS2}, the same authors have proved that black hole spacetimes in 
classical general relativity are characterized by an infinite head of supertranslation hair, in
addition to Arnowitt-Deser-Misner mass $M$, linear momentum $\vec P$, angular momentum
$\vec J$ and boost charge $\vec K$. Classical superrotation charges measured at infinity lead to
distinct black holes. Solutions with supertranslation hair are diffeomorphic to 
Schwarzschild spacetime, and a black hole can be supertranslated by throwing in an asymmetric
shock wave \cite{HPS2}.

Some interesting recent work inspired by Refs. \cite{ST2014,HPS1,HPS2} is the one by
Asorey and coauthors in Ref. \cite{Asorey}. These authors study infrared transformations
in local quantum physics. Their observables are smeared by test functions, at first vanishing
at infinity. They prove that the equations of motion can be seen as constraints, that generate
the group of space- and time-dependent gauge transformations. Infrared non-trivial effects are
then obtained by allowing for test functions that do not vanish at infinity. The resulting
extended operators generate a larger group, and the quotient of the two groups gives rise
to the so-called superselection sectors, that distinguish among the various infrared sectors.
The BMS group plays the role of changing the superselection sector.

As far as black hole theory is concerned, we find it worth bringing to the attention of the reader
the work in Ref. \cite{Bousso}, where the authors have provided a derivation and an observable
realization of the algebra imposed by Strominger on unobservable boundary degrees of freedom. The
conservation laws associated with asymptotic symmetries are found to arise from free propagation 
of infrared modes. In their opinion, this adds evidence in favour of soft charges failing to constrain
the hard scattering problem, and hence not being relevant for the so-called black hole
information paradox.

On reverting now to classical theory, we stress again the new perspective resulting from the
work in Refs. \cite{ST2014,HPS2}: the very existence of a scattering problem from past to future
null infinity in general relativity makes it necessary to consider an infinite number of conserved
supertranslation and superrotation charges. This property was not discovered for several decades,
possibly because many research lines focused on space-times where the peeling property and
the Penrose conformal infinity can be exploited. Thanks to the work of Christodoulou and 
Klainerman \cite{CK1993}, it is by now well known that generic space-times in a finite
neighbourhood of Minkowski lie precisely in the region where it is possible to define an infinite
number of finite, non-vanishing and conserved supertranslation and superrotation charges. 
The work in Ref. \cite{HPS2} has studied a family of space-times whose asymptotic forms lie in
the same region and possess the conserved charges, although they are not necessarily in a
small neighbourhood of Minkowski and may well contain black holes in the interior.
Further work on the rigorous theory of supertranslation and superrotation charges is likely
to shed new light on the mathematical structures of space-time 
\cite{ES1992} and their physical implications.

\section*{acknowledgments}
The authors are grateful to the Dipartimento di Fisica ``Ettore Pancini'' of Federico II University for
hospitality and support. 

\appendix
\chapter{}
\section{Christoffel Symbols}
\label{C}
The metric tensor in \eqref{eqn:136} is 
\begin{equation*}
g_{ab}=\left(\begin{matrix}\frac{V}{r}e^{2\beta}-r^2h_{AB}U^AU^B & e^{2\beta} & r^2h_{2B}U^B & r^2h_{3B}U^B \\
e^{2\beta} & 0 & 0 & 0 \\
r^2h_{2A}U^A & 0 & -r^2h_{22} &-r^2h_{23} \\
r^2h_{3A}U^A & 0 & -r^2h_{32} & -r^2h_{33}
\end{matrix}\right),
\end{equation*}
and its inverse is 
\begin{equation*}
g^{ab}=\left(\begin{matrix} 0 & e^{-2\beta} & 0 & 0 \\
e^{-2\beta} & -\frac{V}{r}e^{-2\beta} & U^2e^{-2\beta} & U^3e^{-2\beta} \\
0 & U^2e^{-2\beta}& -\frac{h_{22}}{r^2} &-\frac{h_{23}}{r^2} \\
0 & U^3e^{-2\beta} & -\frac{h_{32}}{r^2} & -\frac{h_{33}}{r^2}
\end{matrix}\right),
\end{equation*}
while the inverse of the matrix $h_{AB}$ in \eqref{eqn:141} is
\begin{equation*}
h^{AB}=\left(\begin{matrix}
e^{-2\gamma}\cosh 2\delta & -\displaystyle{\frac{\sinh 2\delta}{\sin\theta}}\\ \\
-\displaystyle{\frac{\sinh 2\delta}{\sin\theta}} & \displaystyle{\frac{e^{2\gamma}\cosh 2\delta}{\sin^2\theta}}
\end{matrix}\right).
\end{equation*}
It can be easily verified that 
\begin{equation*}
h^{AB}\partial_r h_{AB}=h^{AB}\partial_u h_{AB}=0.
\end{equation*}
The Christoffel symbols are
\begin{align*}
&\Gamma^{u}{}_{rr}=0,\\
&\Gamma^{r}{}_{rr}=2\partial_r\beta,\\
&\Gamma^{A}{}_{rr}=0,\\
&\Gamma^{u}{}_{rA}=0,\\
&\Gamma^{r}{}_{rA}=\frac{e^{-2\beta}r^2}{2}h_{AB}\left(\partial_r U^B\right)+\partial_{A} \beta,\\
&\Gamma^{B}{}_{rA}=\frac{\delta^B_A}{r}+\frac{\left(\partial_r h_{AC}\right)h^{BC}}{2},\\
&\Gamma^{u}{}_{AB}=e^{-2\beta}rh_{AB}+\frac{e^{-2\beta}r^2}{2}\left(\partial_rh_{AB}\right),\\
&\Gamma^{r}{}_{AB}=\frac{e^{-2\beta}r^2}{2}\left(\partial_AU_B+\partial_BU_A\right)
+\frac{e^{-2\beta}r^2}{2}\left(\partial_uh_{AB}\right)-Ve^{-2\beta}h_{AB}\\
&-\frac{rVe^{-2\beta}}{2}\left(\partial_rh_{AB}\right)-U^C\frac{e^{-2\beta}r^2}{2}\left(\partial_Ah_{CB}
+\partial_Bh_{AC}-\partial_Ch_{AB}\right),\\
&\Gamma^{C}{}_{AB}=rU^Ch_{AB}e^{-2\beta}+\frac{r^2e^{-2\beta}}{2}U^C\left(\partial_rh_{AB}\right)\\
&+\frac{h^{CD}}{2}\left(\partial_Ah_{DB}+\partial_Bh_{DA}-\partial_{D}h_{AB}\right),\\
&\Gamma^{u}{}_{Au}=\partial_A\beta-re^{-2\beta}U_A-\frac{r^2e^{-2\beta}}{2}\left(\partial_rU_A\right),\\
&\Gamma^{r}{}_{Au}=\frac{\partial_AV}{2r}-\frac{r^2e^{-2\beta}}{2}U^B\left(\partial_AU_B\right)
+e^{-2\beta}VU_A+\frac{Vre^{-2\beta}}{2}\left(\partial_rU_A\right)\\&
-\frac{r^2e^{-2\beta}}{2}U^B\left(\partial_uh_{AB}\right)-\frac{r^2e^{-2\beta}}{2}U^B\left(\partial_BU_A\right),\\
&\Gamma^{B}{}_{Au}=U^B\left(\partial_A\beta\right)-re^{-2\beta}U^B\left(\partial_rU_A\right)
-\frac{r^2e^{-2\beta}}{2}U^B\left(\partial_rU_A\right)-\frac{h^{BC}}{2}\left(\partial_AU_C\right)\\
&+\frac{h^{BC}}{2}\left(\partial_CU_A\right)+\frac{h^{BC}}{2}\left(\partial_uh_{AC}\right),\\
&\Gamma^{u}{}_{ru}=re^{-2\beta}U^AU_A+\frac{r^2e^{-2\beta}}{2}U^A\left(\partial_rU_A\right)
-U^A\left(\partial_A\beta\right),\\
&\Gamma^{r}{}_{ru}=\frac{\partial_rV}{2r}-\frac{V}{2r^2}+\frac{V}{r}\left(\partial_r\beta\right)
-\frac{r^2e^{-2\beta}}{2}U^A\left(\partial_rU_A\right)-U^A\left(\partial_A\beta\right),\\
&\Gamma^{A}{}_{ru}=-\frac{U^A}{r}-\frac{h^{AB}}{2}\left(\partial_rU_B\right)+\frac{h^{AB}}{r^2}\left(\partial_B\beta\right).
\end{align*}

\end{document}